\newif\ifcomments
\commentstrue

\documentclass[letterpaper,10pt]{article}
\usepackage[section]{placeins}
\AtBeginDocument{%
  \providecommand\BibTeX{{%
    \normalfont B\kern-0.5em{\scshape i\kern-0.25em b}\kern-0.8em\TeX}}}

\newcommand{\Huntsman}{Huntsman\xspace}
\usepackage{flushend}
\newcommand{\TRAP}{{\sc Trap}\xspace}

\newcommand{\mypar}[1]{\noindent\textbf{#1}}

\usepackage[utf8]{inputenc}
\usepackage[T1]{fontenc}
\usepackage{fixltx2e}
\usepackage{graphicx}
\graphicspath{{./}{images/}}
\usepackage{longtable}
\usepackage{float}
\usepackage{wrapfig}
\usepackage{rotating}
\usepackage[normalem]{ulem}
\usepackage{amsmath}
\usepackage{textcomp}
\usepackage{marvosym}
\usepackage{esvect}
\usepackage{wasysym}
\usepackage{hyperref}
\usepackage{cleveref}
\usepackage{mathtools} % for "\DeclarePairedDelimiter" macro
\DeclarePairedDelimiter{\floor}{\lfloor}{\rfloor}
\DeclarePairedDelimiter{\ceil}{\lceil}{\rceil}
\usepackage{color}
\usepackage{xspace}
\definecolor{olive}{rgb}{0.3, 0.4, .1}
\definecolor{pinegreen}{cmyk}{0.92,0,0.59,0.25}
\usepackage[textwidth=0.15\textwidth,textsize=scriptsize]{todonotes}
\usepackage{pifont} %ding
\usepackage{ dsfont }
\usepackage{algorithm}
\usepackage{algorithmicx}
\usepackage{algpseudocode}
\usepackage{amsthm}
\theoremstyle{plain}
\newtheorem{thm}{Theorem}[section]
\newtheorem{lem}[thm]{Lemma}

\newtheorem*{cor}{Corollary}
\usepackage{macros}
\usepackage{ioa_code_small-1}

\theoremstyle{definition}
\newtheorem{defn}{Definition}[section]

\theoremstyle{remark}

\usepackage{enumitem}
\newcommand{\RNum}[1]{\uppercase\expandafter{\romannumeral #1\relax}}

\usepackage{tikz}
\usetikzlibrary{decorations.pathreplacing,calligraphy}
\tikzstyle{cblue}=[circle, draw, thin,fill=cyan!20, scale=0.8]
\tikzstyle{qgre}=[rectangle, draw, thin,fill=green!20, scale=0.8]
\tikzstyle{rpath}=[ultra thick, opacity=0.4]
\tikzstyle{legend_isps}=[rectangle, rounded corners, thin, 
                       fill=gray!20, text=blue, draw]
                        
\tikzstyle{legend_overlay}=[rectangle, rounded corners, thin,
                           minimum width=2.5cm, minimum height=0.8cm,
                           ]
\tikzstyle{legend_phytop}=[rectangle, rounded corners, thin,
                          minimum width=2.5cm, minimum height=0.8cm,
                          ]
\tikzstyle{legend_general}=[rectangle, rounded corners, thin,
                          top color= white,bottom color=lavander!25,
                          minimum width=2.5cm, minimum height=0.8cm,
                          violet]
\pgfdeclarelayer{background}
\pgfdeclarelayer{foreground}
\pgfsetlayers{background,main,foreground}
\tikzstyle{decision} = [diamond, draw, fill=green!20, 
    text width=4.5em, text badly centered, node distance=3cm, inner sep=0pt]
\tikzstyle{block} = [rectangle, draw, fill=green!20, 
, text centered, rounded corners,align=left]
\tikzstyle{block2} = [rectangle, fill=green!20, 
, text centered, rounded corners,align=left]
\tikzstyle{line} = [draw, -latex']
\tikzstyle{cloud} = [draw, ellipse,fill=purple!20, node distance=3cm,
    minimum height=2em]

\usetikzlibrary{shapes,arrows,positioning,automata,calc,patterns}
\tikzset{node/.style={circle,draw,inner sep=0pt, minimum size=2em}}
\usetikzlibrary{external}
\usetikzlibrary{positioning}
\usetikzlibrary{fit}
\usetikzlibrary{backgrounds}
\usetikzlibrary{calc}
\usetikzlibrary{mindmap}
\usetikzlibrary{decorations.text}
\usetikzlibrary{arrows}

\tikzset{
  treenode/.style = {align=center, inner sep=0pt, text centered,
    font=\sffamily},
  arn_n/.style = {treenode, circle, white, font=\sffamily\bfseries, draw=black,
    fill=black, text width=1.5em},% arbre rouge noir, noeud noir
  arn_r/.style = {treenode, circle, red, draw=red, 
    text width=1.5em, very thick},% arbre rouge noir, noeud rouge
  arn_x/.style = {treenode, rectangle, draw=black,
    minimum width=0.5em, minimum height=0.5em}% arbre rouge noir, nil
}
\usetikzlibrary{shapes.multipart,shapes.geometric}
\tikzset{main node/.style={circle,fill=blue!20,draw,inner sep=0pt, minimum size=2em},
}
\tikzset{node message/.style={circle split,fill=blue!20,draw,inner sep=0pt},
}
\tikzset{node home/.style={star,star points=7,star point ratio=0.8,fill=green!20,draw,minimum size=0.5cm,inner sep=0pt},}
\tikzset{man op/.style={ % requires library shapes.geometric
        draw,
        trapezium,
        shape border rotate=180,
        text width=2cm,
        align=center,
      },}
    
\usetikzlibrary{positioning}
\usetikzlibrary{fit}
\usetikzlibrary{backgrounds}
\usetikzlibrary{calc}
\usetikzlibrary{shapes}
\usetikzlibrary{mindmap}
\usetikzlibrary{decorations.text}
\usetikzlibrary{arrows}
\usetikzlibrary{shapes,shapes.multipart}
\usetikzlibrary{chains}
\tikzset{main node/.style={circle,fill=blue!20,draw,minimum size=1cm,inner sep=0pt},
}
\tikzset{node message/.style={circle split,fill=blue!20,draw,minimum size=0.5cm,inner sep=0pt},
}
\tikzset{node home/.style={star,star points=7,star point ratio=0.8,fill=green!20,draw,minimum size=0.5cm,inner sep=0pt},}

\usepackage{epstopdf}

\makeatletter
\newcommand*{\inlineequation}[2][]{%
  \begingroup
    \refstepcounter{equation}%
    \ifx\\#1\\%
    \else
      \label{#1}%
    \fi
    \relpenalty=10000 %
    \binoppenalty=10000 %
    \ensuremath{%
      #2%
    }%
    ~\@eqnnum
  \endgroup
}
\makeatother

\begin{document}
\newcommand{\arcomm}[1]{\todo[color=green,bordercolor=black,linecolor=black]{\textsf{\scriptsize\linespread{1}\selectfont ARP: #1}}}
\newcommand{\arcommin}[1]{\todo[inline,color=green,bordercolor=green,linecolor=green]{\textsf{ARP: #1}}}
\newcommand{\system}{[system name] }

\newcommand{\boxedtext}[1]{\fbox{\scriptsize\bfseries\textsf{#1}}}

\newcommand{\greenremark}[2]{
   \textcolor{pinegreen}{\boxedtext{#1}
      {\small$\blacktriangleright$\emph{\textsl{#2}}$\blacktriangleleft$}
    }}

  \definecolor{burntorange}{rgb}{0.8, 0.33, 0.0}
  \definecolor{blue}{rgb}{0.0, 0.0, 0.5}
  \newcommand{\changeremark}[2]{
   \textcolor{burntorange}{\boxedtext{#1}
      {\small$\blacktriangleright$\emph{\textsl{#2}}$\blacktriangleleft$}
}}
\newcommand{\myremark}[2]{
   % \textcolor{red}{\boxedtext{#1}
   %    {\small$\blacktriangleright$\emph{\textsl{#2}}$\blacktriangleleft$}
   %  }}
      {  \lowercase{\color{blue}}\boxedtext{#1}
      {\small$\blacktriangleright$\emph{\textsl{#2}}$\blacktriangleleft$}
    }}
  \newcommand{\myremarknew}[2]{
   \textcolor{violet}{\boxedtext{#1}
      {\small$\blacktriangleright$\emph{\textsl{#2}}$\blacktriangleleft$}
}}
\newcommand{\NewRemark}[2]{
   \textcolor{violet}{
      {\small$\blacktriangleright$\emph{\textsl{#2}}$\blacktriangleleft$}
}}

\newcommand{\redremark}[2]{
   \textcolor{red}{\boxedtext{#1}
      {\small$\blacktriangleright$\emph{\textsl{#2}}$\blacktriangleleft$}
    }}
  
\ifcomments
  \newcommand\ARP[1]{\myremark{ARP}{#1}}
  \newcommand\ARPN[1]{\myremarknew{A}{#1}}
  \newcommand\vincent[1]{{\color{red}{VG: #1}}}
  \newcommand\vincentR[2]{{\color{red}{#1 \sout{#2}}}}
  \newcommand{\warning}[1]{\redremark{\fontencoding{U}\fontfamily{futs}\selectfont\char 66\relax}{#1}}
  \newcommand\NEW[1]{\NewRemark{NEW}{#1}}
  \newcommand\TODO[1]{\greenremark{TODO}{#1}}
  \newcommand\CHANGE[1]{\changeremark{CHANGE(?)}{#1}}
\else
  \newcommand\ARP[1]{}
  \newcommand\ARPN[1]{}
  \newcommand\vincent[1]{}
  \newcommand{\warning}[1]{}
  \newcommand\NEW[1]{#1}
  \newcommand\TODO[1]{}
  \newcommand\CHANGE[1]{}
\fi

\newcommand{\gametype}{Consensus\xspace}
\newcommand{\unavoidable}{one-way bait\xspace}
\newcommand{\ragreement}{relaxed agreement\xspace}
\newcommand{\problem}{rational agreement\xspace}
\newcommand{\Pproblem}{Rational agreement\xspace}
\newcommand{\Problem}{Rational Agreement\xspace}
\newcommand{\dfraud}{direct-fraud\xspace}
% <<<<<<< HEAD
% \newcommand{\zeroloss}{loss-free reward\xspace}
% =======
\newcommand{\zeroloss}{lossfree-reward\xspace}
% >>>>>>> a6f2ec9839aab1c88453046815d545a1604b4def
\newcommand{\scheduler}{scheduler\xspace}
\newcommand{\rdominance}{baiting-dominance\xspace}
\newcommand{\baitingagreement}{baiting-agreement\xspace}
\newcommand{\strong}{strong\xspace}
\newcommand{\interim}{interim\xspace}
\newcommand{\Interim}{Interim\xspace}
\newcommand{\statement}{Let $n$ players execute the \TRAP protocol, out
of which there can be a coalition of $k$ rational players and $t$
Byzantine players such that $n>2(k+t)$ and $n>\frac{3}{2}k+3t$. }
\newcommand{\statementprops}{Suppose
  $m(k,t)=\floor{\frac{k+t-n}{2}+t}+1$ rational players in the coalition play the baiting strategy. Then,\xspace}

% \title{Agreement in the Presence of Disagreeing Rational Players:\\The \TRAP Protocol}
% \title{On the power of baiting for the rational agreement problem:\\The \TRAP Protocol}
% \title{\TRAP: Rational Agreement by Baiting Byzantine Participants}
\title{\TRAP: The Bait of Rational Players to Solve Byzantine Consensus}
%Rational Agreement despite Byzantine Participants}
% On the power of baiting to turn rational players into Byzantine whistleblowers
% \title{The TRAP protocol: turning disagreeing rational players into Byzantine whistleblowers through}
% \title{The TRAP protocol, or how baiting strategies can turn disagreeing rational players
%   into Byzantine whistleblowers}
% \title{On the power of baiting strategies to turn disagreeing rational players
%   into Byzantine whistleblowers: the TRAP protocol}

% \author{{\rm Alejandro Ranchal-Pedrosa}\\
%   \small{University of Sydney} \\
% %\textit% name of organization (of Aff.)\\
%   \small{Sydney, Australia} \\
%   \small{alejandro.ranchalpedrosa@sydney.edu.au}
%   \and
%   {\rm Vincent Gramoli}\\
%   \small{University of Sydney}\\
%   \small{Sydney, Australia} \\
%   \small{vincent.gramoli@sydney.edu.au}}
\author{{\rm Alejandro Ranchal-Pedrosa}\\
  \small{University of Sydney} \\
%\textit% name of organization (of Aff.)\\
  \small{Sydney, Australia} \\
  \small{alejandro.ranchalpedrosa@sydney.edu.au}
  \and
  {\rm Vincent Gramoli}\\
  \small{University of Sydney and EPFL}\\
  \small{Sydney, Australia} \\
  \small{vincent.gramoli@sydney.edu.au}}

\maketitle

\begin{abstract}
% \vincentR{
%When targeting classic Byzantine fault tolerance, 
%The Byzantine vs. correct failure model led to well-known impossibility results. 
%With Byzantine participants, 
It is impossible to solve the Byzantine consensus problem in an open network of $n$ participants
%in open networks 
if only $2n/3$ or less of them are correct.
%with only $2n/3$ correct participants or less.
As blockchains need to solve consensus, one might 
%thus omit that blockchain users can be rational and 
think that blockchains need more than $2n/3$ correct participants.
But it is yet unknown whether consensus can be solved when less than $2n/3$ participants are correct and $k$ participants are rational players, which misbehave if they can gain the loot. 
% , even though 
% some nodes are rational players
%With the recent advent of blockchains, a more realistic model is represented as rational players, 
% that misbehave if they are guaranteed to gain the loot.
%But it is yet unknown whether consensus can be solved when less than $2n/3$ are correct participants and $k$ of them are rational players. 
Trading correct participants for rational players may not seem helpful to solve consensus
%It remains unclear weather trading correct participants for rational players could help solving consensus, 
since rational players can misbehave whereas correct participants, by definition, cannot.
%Intuitively, trading correct participants for rational players should not help solving consensus as rational players can misbehave where correct, by definition, do not.

In this paper, we show that consensus is actually solvable in this model, even with less than $2n/3$ correct participants.
The key idea is a \emph{baiting strategy} that lets rational players pretend to misbehave in joining a coalition but rewards them to betray this coalition before the loot gets stolen.
We propose \TRAP, a protocol that builds upon recent advances in the theory of accountability to
solve consensus as soon as $n>\max\bigl(\frac{3}{2}k+3t,2(k+t)\bigr)$:
%, solving consensus for example among $n=5$ players with $k=1$ rational player, $t=1$ Byzantine player and $3$ correct players.
by assuming that private keys cannot be forged, this protocol is an equilibrium where no coalition of $k$ rational players can coordinate to increase their expected utility 
regardless of the arbitrary behavior of up to $t$ Byzantine players.
% }
% {
% In this paper, a novel Byzantine consensus protocol among $n$ players is proposed for the partially synchronous model. In particular, by assuming that standard 
% cryptography is unbreakable, and that $n>\max\bigl(\frac{3}{2}k+3t,2(k+t)\bigr)$, this protocol is an equilibrium where no coalition of $k$ rational players can coordinate to increase their expected utility 
% regardless of the arbitrary behavior of up to $t$ Byzantine players.}
%this protocol is a $(k,f)$-robust equilibrium \vincent{still not sure if this shouldn't be $k+f$-resilient} equilibrium in that it tolerates the coordinated changes of strategy of $k$ rational players and $t$ Byzantine players \TODO{remove?: and is $t$-immune in that it guarantees that $n-t$ correct processes agree despite the arbitrary behaviors of $f\leq t$ Byzantine players}. 

Finally, we show that a baiting strategy is necessary and sufficient to solve this, so-called \emph{\problem} problem. 
First, we show that it is impossible to solve this \problem %$(k,f)$-robust 
problem without
implementing a baiting strategy. Second, 
%we describe how the \TRAP protocol solves the \problem problem by building upon recent advances in the theory of 
%accountability. 
the existence of \TRAP demonstrates the sufficiency of the baiting strategy.
Our \TRAP protocol finds applications in blockchains
to prevent players from disagreeing, that could otherwise lead to ``double spending''.
%where players are incentivized to steal assets by leading other players to a disagreement on two distinct decisions where they ``double spend''.
%
%$k$ are rational players and up to $t$ are Byzantine players is proposed.
%Some work proposed ways to bridge the gap between game theory and distributed computing. 
%The original idea was to make a Nash equilibrium resilient to a coalition of $k$ rational players 
%capable of coordinating their strategy deviation and recent results have proposed to consider $t$ byzantine players as irrational and reach a Nash equilibrium among $n-t$ rational players.
%Unfortunately, the only solutions to reach consensus tolerate crash failures only of electing a leader without failures or 
%ensuring a fair choice in the decision in a synchronous model. We present a protocol that converges \vincent{Do we show convergence actually?} \ARP{not sure what convergence means here} towards an equilibrium where agreement is reached and that is resilient to a coalition of $k$ rational players and immune to $t$ byzantine players.
\end{abstract} 

% \begin{CCSXML}
% <ccs2012>
%   <concept>
%       <concept_id>10002978.10002986.10002989</concept_id>
%       <concept_desc>Security and privacy~Formal security models</concept_desc>
%       <concept_significance>500</concept_significance>
%       </concept>
%   <concept>
%       <concept_id>10002978.10003006.10003013</concept_id>
%       <concept_desc>Security and privacy~Distributed systems security</concept_desc>
%       <concept_significance>500</concept_significance>
%       </concept>
%  </ccs2012>
% \end{CCSXML}
% \ccsdesc[500]{Security and privacy~Formal security models}
% \ccsdesc[500]{Security and privacy~Distributed systems security}

% \keywords{Blockchain, consensus, game theory, robustness, fault tolerance}
% \keywords{agreement, consensus, game theory, robustness, fault tolerance} %TODO mandatory; please add comma-separated list of keywords
% \CopyrightYear{2022}
% \conferenceinfo{ASIA CCS '22,}{May 30-June 3, 2022, Nagasaki, Japan}
% \isbn{978-1-4503-9140-5/22/05}\acmPrice{$15.00}
% \doi{https://doi.org/10.1145/3488932.3517386}

\maketitle 
% \tableofcontents

\section{Introduction and Background}
\label{sec:intro}
% context: what is a coalition, deviation, mediator, Nash equilibrium
Consider $n$ players, each with some initial value. Solving the Byzantine consensus problem consists of designing a protocol guaranteeing that the $n-t$ non-Byzantine players agree by outputting the same value and despite the presence of %initially undistinguishable 
up to $t$ arbitrarily behaving Byzantine players. 
% accountability
Standard cryptography, which permits oblivious signed transfers and assumes computationally bounded players, has recently been used to build undeniable proofs identifying
%of the identity of 
the players that led the system to a disagreement~\cite{civit2020brief,CGG21}. 
Although this construction has not been used in the game theoretical context, one can intuitively see its application to implement a \emph{baiting strategy}
%one can easily been applied 
that incentivizes rational players to solve blockchain consensus:
the idea is to reward rational players to pretend
%one can reward rational players for pretending
to join a coalition in order to deceit the coalition by exposing undeniable 
\emph{proofs-of-fraud} (PoFs).

% k-robustness
Placed in the game theory context, one can see a consensus 
protocol among rational players as a Nash equilibrium, however, a Nash equilibrium only prevents one rational player from increasing its utility by deviating solely, but it fails at preventing multiple rational players from increasing their utility by colluding and by all deviating together. 
The resilience against this type of coalition, mentioned originally in the 50s~\cite{Au59}, is needed to solve consensus despite a coalition of $k$ players.
% punishment strategy 
Ben-Porath~\cite{Ben03} shows that one can simulate a Nash equilibrium with a central trusted mediator provided that there is a punishment strategy
%for players to threaten potential deviators. 
%A punishment strategy represents a threat that correct and rational
to threaten rational players in case they deviate and 
Heller~\cite{Hel10} strengthens Ben-Porath's result to allow coalitions. Abraham et al.~\cite{ADGH} applied these results to secret sharing in the fully distributed setting, by showing that one %can assume private channels or if we
can simulate a mediator with cheap talks and 
assuming the same standard cryptography we assume.
%thats permit oblivious transfer and computationally bounded players.
%(that imply oblivious transfer and computationally bounded players), we can simulate any equilibrium of a game with a mediator provided that there is a punishment strategy.
A $(k,t)$-punishment strategy guarantees that if up to $k$ rational players deviate, then more than $t$ non-deviating players, by playing the punishment strategy, can lower the utility of these rational players.

% t-immunity
Another challenge when making consensus Byzantine fault tolerant is for the equilibrium to be immune 
to $t$ Byzantine players that act arbitrarily or whose utility functions 
are unknown.
%Another requirement when solving consensus despite $t$ Byzantine players would consist of designing a game among $n$ players such that there exists a unique Nash equilibrium among $n-t$ players despite the presence of $t$ irrational players.
%
Abraham et al.~\cite{ADGH} were the first to
formalize %coalitions of up to $k$ rational players and $t$ Byzantine players, in what they
%described as 
$k$-resilience, $t$-immunity and $\epsilon$-$(k,t)$-robustness.
A protocol is a $k$-resilient equilibrium if no rational coalition of size up to $k$
can increase their utility by deviating in a coordinated way.
%While $t$-immunity describes a protocol that tolerates coalitions of
%up to $t$ Byzantine players and $k$-resilience tolerates coalitions of up to
%$k$ rational players~\cite{ADGH}, an $\epsilon$-$(k,t)$-robust protocol tolerates coalitions of up to
%$k$ rational players and $t$ Byzantine players.
A protocol is $t$-immune if the expected utility of the $n-t$ non-faulty players is not decreased regardless of the arbitrary behavior of up to $t$ Byzantine players.
A protocol is $\epsilon$-$(k,t)$-robust if no coalition of $k$ rational players can coordinate to increase their expected utility by $\epsilon$
regardless of the arbitrary behavior of up to $t$ Byzantine players, even if the Byzantine players join their coalition, where $\epsilon$ accounts for the (small) probability of breaking the cryptography.

Considering rational players alongside the well-known Byzantine faults to ensure the agreement property of the consensus problem allows us to evaluate protocols against coalitions of players that are incentivized to deviate for their own benefit or to break the system, rather than just because of their faulty nature. 
This \emph{rational agreement} problem finds application in blockchains where participants can be modeled as players incentivized to gain valuable cryptocurrency assets.
In 2020, these assets incentivized players to deviate from their blockchain protocol by forcing a disagreement to fork the blockchain, leading to a double spending of US$\mathdollar 70,000$ in Bitcoin Gold\footnote{\href{https://news.bitcoin.com/bitcoin-gold-51-attacked-network-loses-70000-in-double-spends/}{https://news.bitcoin.com/bitcoin-gold-51-attacked-network-loses-70000-in-double-spends/}} and US$\mathdollar 5.6$ million in Ethereum Classic\footnote{\href{https://news.bitcoin.com/5-6-million-stolen-as-etc-team-finally-acknowledge-the-51-attack-on-network/}{https://news.bitcoin.com/5-6-million-stolen-as-etc-team-finally-acknowledge-the-51-attack-on-network/}}.
Solving rational agreement is key to avoid these losses because it ensures that players agree on a unique block at each index of the chain, thus preventing these forks from being exploited to double spend.  
%
%Our work thus finds relevance in the context of blockchains
% where a disagreement can allow a player to steal digital assets but where 
% %requiring  
% the asset deposit of some players can be used to threaten them.
% %players to deposit some of their assets can be used to threaten them. 
% This rational behavior has already materialized in the recent losses of $\mathdollar 70,000$\footnote{\href{https://news.bitcoin.com/bitcoin-gold-51-attacked-network-loses-70000-in-double-spends/}{https://news.bitcoin.com/bitcoin-gold-51-attacked-network-loses-70000-in-double-spends/}} and $\$18$ million\footnote{\href{https://news.bitcoin.com/bitcoin-gold-hacked-for-18-million/}{https://news.bitcoin.com/bitcoin-gold-hacked-for-18-million/}} in Bitcoin Gold, and of $\mathdollar 5.6$ million in Ethereum Classic\footnote{\href{https://news.bitcoin.com/5-6-million-stolen-as-etc-team-finally-acknowledge-the-51-attack-on-network/}{https://news.bitcoin.com/5-6-million-stolen-as-etc-team-finally-acknowledge-the-51-attack-on-network/}}, justifying the importance of the evaluation we present in this work.
% %if no coalition with up to $k$ rational players and $t$ Byzantine players can increase the utility of the $k$ rational players by deviating.
% %\ARP{remove next sentence, not true} 
% %A protocol is $(k,t)$-robust if it is a $k$-resilient equilibrium and is $t$-immune, we say that it is $\epsilon$-$(k,t)$-robust to cope with the small
% %probability that the coalition breaks the cryptography.
Interestingly, some blockchains already require participants to deposit cryptocurrency assets in the form of staking and these
deposits could be used by the consensus protocol to incentivize their owners to behave. Unfortunately, these blockchains fork in the presence of rational players and we are not aware of any solution to this rational agreement problem.

\subsection{Our result}
% contribution
In this paper, we show that a \emph{baiting strategy}, that incentivize rational players to bait deviating players into a trap,
%, that one can implement with undeniable proofs-of-fraud (PoFs), 
is necessary and sufficient
to solve this \emph{rational agreement} problem by offering a consensus protocol that is
robust to a coalition of up to $k$ rational players and $t$ Byzantine players.
Our first contribution is thus to formalize the notion of a baiting strategy and show that this baiting strategy is 
necessary to solve the rational agreement problem.
%as a particular case of punishment strategy 
% and show that it is impossible to solve the 
% rational agreement problem without a baiting strategy, a strategy that enough rational players of the coalition prefer to play rather than colluding with the coalition.
Motivated by open networks, where blockchains typically operate, and where the synchrony assumption~\cite{DLS88} in unrealistic, we then solve the rational agreement problem without assuming a known bound on the delay of every message. 
As solving the consensus problem is impossible with complete asynchrony~\cite{FLP85}, we consider the partially synchronous model where there is an unknown bound on the delay of messages~\cite{DLS88}.
%we propose \TRAP, a rational agreement protocol that simply assumes partial synchrony, where there is an unknown bound on the delay of messages~\cite{DLS88}.
This solution bypasses the well-known requirement of $2n/3$ correct participants~\cite{LSP82}, by solving 
consensus when $n>\max(\frac{3}{2}k+3t,2(k+t))$. 
For example if $n=7$ players, then our solution solves consensus with only 4 correct participants, hence tolerating $k=1$
rational player and $t=2$ Byzantine players.

We implement this solution in a new protocol, called \TRAP (\textbf{T}ackling \textbf{R}ational \textbf{A}greement through \textbf{P}ersuasion),
that is $\epsilon$-$(k,t)$-robust when $n>\max(\frac{3}{2}k+3t,2(k+t))$.
\TRAP rewards a single player to expose its coalition by generating proofs-of-fraud with the 
Polygraph accountable consensus protocol~\cite{CGG21}.
%If the reward is larger than the individual return that a rational player gains from 
%causing a disagreement, then rational players in a coalition can find that the
%strategy to form a coalition and cause a disagreement is strictly
%dominated by the strategy to betray the coalition in the extensive-form
%game.
Making initially colluding players decide on whether to betray the coalition regardless of the coalition behavior
is analogous to reducing the extensive-form game into a normal-form game for this particular decision.
In addition, if the reward for exposing the coalition is greater than the individual payoff for causing a disagreement, 
rational players find that the coalition strategy is strictly dominated by the baiting strategy 
in the extensive-form game. This shows similarities with the prisoner's dilemma, in which all rational players prefer to 
betray the coalition than to collude.

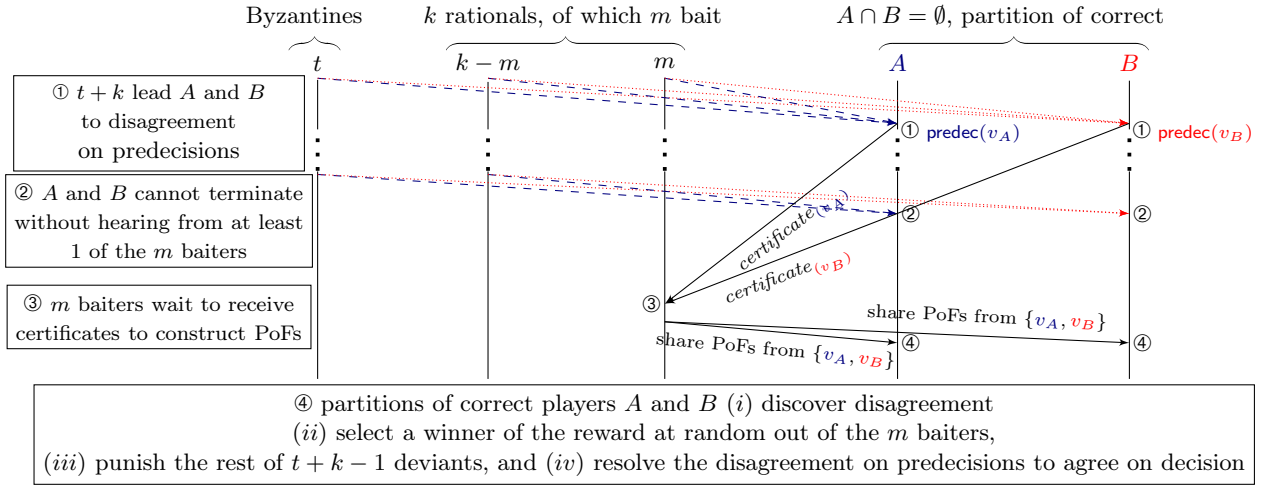
\begin{figure}[tp]
  \hspace{-6em}
    \pgfsetlayers{background,main,foreground}
    \begin{tikzpicture}[node distance=1cm,auto,>=stealth']
      \tikzstyle{every node}=[font=\small]
      \node[] (byz) {$t$};
      % \node[yshift=1.5em,xshift=-1.5em] (byzauxbyz) {Byzantines};
   \node[right = of byz,xshift=1.5em] (km) {$k-m$};
   \node[right = of km,xshift=1.5em] (m) {$m$};
   \node[right = of m,xshift=4.5em] (a) {$\color{blue}A$};
   \node[right = of a,xshift=4.5em] (b) {$\color{red}B$};
   \draw [decorate, 
    decoration = {calligraphic brace,
        raise=5pt,
        amplitude=5pt}] (a.west) --  (b.east)
      node[pos=0.45,above=10pt,black]{$A\cap B=\emptyset$, partition of correct};

      \draw [decorate, 
    decoration = {calligraphic brace,
        raise=5pt,
        amplitude=5pt}] (byz.west)+(-0.2,0) -- (0.3,0)+(byz.east)
      node[pos=0.30,above=10pt,black]{Byzantines};

      \draw [decorate, 
      decoration = {calligraphic brace,
        raise=5pt,
        amplitude=5pt}] (km.west) -- (m.east)
      node[pos=0.53,above=10pt,black,align=center]{$k$ rationals, of which $m$ bait};

   \def\vara{4}
   \node[below of=byz, node distance=\vara cm] (byz_ground) {};
   \node[below of=km, node distance=\vara cm] (km_ground) {};
   \node[below of=m, node distance=\vara cm] (m_ground) {};
   \node[below of=a, node distance=\vara cm] (a_ground) {};
   \node[below of=b, node distance=\vara cm] (b_ground) {};
   
   \def\varb{0.22}
   \draw (byz) -- ($(byz)!\varb!(byz_ground)$);
   \draw (km) -- ($(km)!\varb!(km_ground)$);
   \draw (m) -- ($(m)!\varb!(m_ground)$);
   \draw (a) -- ($(a)!\varb!(a_ground)$);
   \draw (b) -- ($(b)!\varb!(b_ground)$);      

   \def\varc{0.36}
   \def\vard{1.05}
   \draw ($(byz)!\varc!(byz_ground)$) -- ($(byz)!\vard!(byz_ground)$);
   \draw ($(km)!\varc!(km_ground)$) -- ($(km)!\vard!(km_ground)$);
   \draw ($(m)!\varc!(m_ground)$) -- ($(m)!\vard!(m_ground)$);
   \draw ($(a)!\varc!(a_ground)$) -- ($(a)!\vard!(a_ground)$);
   \draw ($(b)!\varc!(b_ground)$) -- ($(b)!\vard!(b_ground)$);

   % \path  ($(byz)!0.22!(byz_ground)$)-- ($(byz)!0.30!(byz_ground)$)node [red, font=\Huge, midway, center] {$\dots$};
   \def\vare{0.25}
   \def\varf{0.36}
   \draw[loosely dotted, line width=0.5mm] ($(byz)!\vare!(byz_ground)$)-- ($(byz)!\varf!(byz_ground)$);
   \draw[loosely dotted, line width=0.5mm] ($(km)!\vare!(km_ground)$)-- ($(km)!\varf!(km_ground)$);
   \draw[loosely dotted, line width=0.5mm] ($(m)!\vare!(m_ground)$)-- ($(m)!\varf!(m_ground)$);
   \draw[loosely dotted, line width=0.5mm] ($(a)!\vare!(a_ground)$)-- ($(a)!\varf!(a_ground)$);
   \draw[loosely dotted, line width=0.5mm] ($(b)!\vare!(b_ground)$)-- ($(b)!\varf!(b_ground)$);
   
   % % \draw (m) -- (m_ground);
   % % \draw (a) -- (a_ground);
   % % \draw (b) -- (b_ground);

   \def\varg{0.05}
   \def\varh{0.2}
   \draw[-latex',color=blue,dashed] ($(byz)!\varg!(byz_ground)$) -- node[below,near start]{% \scriptsize $\lit{dec}(v_A)$
   } ($(a)!\varh!(a_ground)$);
   \draw[-latex',color=red,densely dotted] ($(byz)!\varg!(byz_ground)$) -- node[below,near start]{% \scriptsize $\lit{dec}(v_B)$
   } ($(b)!\varh!(b_ground)$);

   \draw[-latex',color=blue,dashed] ($(km)!\varg!(km_ground)$) -- node[below,near start]{% \scriptsize $\lit{dec}(v_A)$
   } ($(a)!\varh!(a_ground)$);
   \draw[-latex',color=red,densely dotted] ($(km)!\varg!(km_ground)$) -- node[below,near start]{% \scriptsize $\lit{dec}(v_B)$
   } ($(b)!\varh!(b_ground)$);

   \draw[-latex',color=blue,dashed] ($(m)!\varg!(m_ground)$) -- node[below,near start]{% \scriptsize $\lit{dec}(v_A)$
   } ($(a)!\varh!(a_ground)$);
   \draw[-latex',color=red,densely dotted] ($(m)!\varg!(m_ground)$) -- node[below,near start]{% \scriptsize $\lit{dec}(v_B)$
   } ($(b)!\varh!(b_ground)$);

   \def\vari{0.23}
   \def\varj{0.20}
   \def\vark{0.37}
   \def\varm{0.5}
   \node[xshift=2.4em] (a_predec) at ($(a)!\vari!(a_ground)$) {\ding{192~}{\scriptsize \color{blue}$\lit{predec}(v_A)$}};
   \node[xshift=2.4em] (b_predec) at ($(b)!\vari!(b_ground)$) {\ding{192~}{\scriptsize \color{red}$\lit{predec}(v_B)$}};
   \node[xshift=-6em,align=center,style={rectangle,draw},minimum width=11em] (exp1) at ($(byz)!\varj!(byz_ground)$) {\footnotesize \ding{192}~$t+k$ lead $A$ and $B$\\\footnotesize to disagreement\\ on predecisions};
      % \draw[loosely dotted, line width=0.5mm] ($(b)!0.235!(b_ground)$)-- ($(b)!0.285!(b_ground)$);

   \draw[-latex',color=blue,dashed] ($(byz)!\vark!(byz_ground)$) -- node[below,near start]{% \scriptsize $\lit{dec}(v_A)$
   } ($(a)!\varm!(a_ground)$);
   \draw[-latex',color=red,densely dotted] ($(byz)!\vark!(byz_ground)$) -- node[below,near start]{% \scriptsize $\lit{dec}(v_B)$
   } ($(b)!\varm!(b_ground)$);

   \draw[-latex',color=blue,dashed] ($(km)!\vark!(km_ground)$) -- node[below,near start]{% \scriptsize $\lit{dec}(v_A)$
   } ($(a)!\varm!(a_ground)$);
   \draw[-latex',color=red,densely dotted] ($(km)!\vark!(km_ground)$) -- node[below,near start]{% \scriptsize $\lit{dec}(v_B)$
   } ($(b)!\varm!(b_ground)$);
   \node[align=center,xshift=0.5em] (dingaux1) at ($(a)!\varm!(a_ground)$) {\ding{193}};
   \node[align=center,xshift=0.5em] (dingaux1) at ($(b)!\varm!(b_ground)$) {\ding{193}};
   \ding{193}~
   % % \node[xshift=-4em,align=center] (exp1) at ($(byz)!0.09!(byz_ground)$) {$t+k$ lead to\\ disagreement\\ on predecisions};
   % % \node[align=center] (exp1) at ($(m)!0.25!(m_ground)$) {\Huge$\vdots$};

   \def\varn{0.5}
   \def\varo{0.8}
   \node[xshift=-6em,align=center,style={rectangle,draw},minimum width=11em,yshift=-0.28em] (exp2) at ($(byz)!\varn!(byz_ground)$) {\ding{193}~\footnotesize $A$ and $B$ cannot terminate\\\footnotesize without hearing from at least\\\footnotesize $1$ of the $m$ baiters};

   \draw[-latex'] ($(a)!\varj!(a_ground)$) -- node[below,pos=0.5, sloped]{ \scriptsize $\ms{certificate}_{{\color{blue}(v_A)}}$
   } ($(m)!\varo!(m_ground)$);

   \draw[-latex'] ($(b)!\varj!(b_ground)$) -- node[below,near end, sloped]{ \scriptsize $\ms{certificate}_{{\color{red}(v_B)}}$
   } ($(m)!\varo!(m_ground)$);

   \node[align=center,xshift=-0.5em] (dingaux2) at ($(m)!\varo!(m_ground)$) {\ding{194}};

   \node[xshift=-6em,yshift=-0.5em,align=center,style={rectangle,draw},minimum width=11em] (exp3) at ($(byz)!\varo!(byz_ground)$) {\ding{194}~\footnotesize $m$ baiters wait to receive\\\footnotesize certificates to construct PoFs};

   \def\varp{0.85}
   % \node[align=center,yshift=0.2em,xshift=0.2em] (mpofs) at ($(m)!\varp!(m_ground)$) {$\ms{pofs}_{({\color{blue}v_A},{\color{red}v_B})}$};

   \def\varq{0.86}
   \def\varr{0.93}
   \draw[-latex'] ($(m)!\varq!(m_ground)$) -- node[below,pos=0.5, sloped,align=center]{ \scriptsize share PoFs from $\{{\color{blue}v_A},{\color{red}v_B}\}$
   } ($(a)!\varr!(a_ground)$);

   \draw[-latex'] ($(m)!\varq!(m_ground)$) -- node[above,pos=0.7, sloped,align=center]{ \scriptsize share PoFs from $\{{\color{blue}v_A},{\color{red}v_B}\}$
   } ($(b)!\varr!(b_ground)$);

   \node[align=center,xshift=0.5em] (dingaux2) at ($(a)!\varr!(a_ground)$) {\ding{195}};
   \node[align=center,xshift=0.5em] (dingaux2) at ($(b)!\varr!(b_ground)$) {\ding{195}};
   
   \node (aux) at ($(byz)!0.5!(b)$) {};
   \node (aux_ground) at ($(byz_ground)!0.5!(b_ground)$) {};
   \node[align=center,style={rectangle,draw}, minimum width=35em,xshift=-3em,yshift=-0.4em] at ($(aux)!1.2!(aux_ground)$){\ding{195}~partitions of correct players $A$ and $B$ $(i)$ discover disagreement\\ $(ii)$ select a winner of the reward at random out of the $m$ baiters,\\$(iii)$ punish the rest of $t+k-1$ deviants, and $(iv)$ resolve the disagreement on predecisions to agree on decision};
 \end{tikzpicture}
 \caption{Example execution of the \Huntsman protocol. First, \ding{192}~ all $t$ Byzantine and $k$ rational players collude to cause a disagreement on the output of the accountable consensus protocol, resulting in $A$ and $B$ predeciding different outputs. Then, \ding{193}~ $m$ of the $k$ rational players decide to bait while executing the BFTCR protocol, preventing $A$ and $B$ from deciding their disagreeing predecisions. As such, \ding{194}~the $m$ baiters wait until they receive proof of the disagreement on predecisions, to then \ding{195}~prove the disagreement committing to and revealing the proofs-of-fraud in the BFTCR protocol. Hence, neither $A$ nor $B$ decide their conflicting predecisions, but instead reward one of the $m$ baiters, punish the rest of $t+k-1$ players responsible for the disagreement on predecisions, and resolve the disagreement, deciding one of $v_A$ or $v_B$, or, depending on the application, merging both.}
 \label{fig:diag}
  \end{figure}

More specifically, our protocol ``pre-decides'' the decisions from Polygraph~\cite{CGG21} that it extends with 
the \emph{Byzantine Fault Tolerant Commit-Reveal protocol (BFTCR)}, which consists of two reliable
broadcasts and one additional broadcast. As we show in
Figure~\ref{fig:diag}, by offering a reward the \TRAP protocol can
convince $m$ rational players to betray the coalition after they
helped cause a disagreement on predecisions. First, the coalition exposes itself
causing a disagreement on predecisions. Second, the $m$ rational players
from the coalition that decide to betray are enough to pause
termination of the BFTCR protocol. Third, these players can wait to
gather enough PoFs of the disagreement on predecisions, which they
will get by the property of accountability (since they caused a
disagreement on the output of the Polygraph
protocol). Fourth and finally, once they get PoFs to prove the disagreement to correct players and get the reward, these $m$ players
commit and reveal the PoFs by sharing them during the BFTCR protocol,
after which one of them will be selected at random to get the
reward. We detail further this example in Appendix~\ref{sec:exfig}.

Adding this BFTCR phase ensures
the existence of a baiting strategy (\rdominance) and that the protocol still solves agreement even after playing the
baiting strategy (\baitingagreement).
We also add an additional property, \zeroloss, which states that the
increase in utility for baiting rational players comes at no cost to
non-deviating players. For this purpose, we introduce a deposit per
player, so that the system can always pay the
reward by taking the deposits of the proven coalition at no cost for
non-deviating players.

  \subsection{Related work} \label{sec:relwork}
%% vincent
Considering fault tolerant distributed protocols as games requires to 
cope with a mixture of up to $k$ rational players and $t$ faulty players. 
The idea of mixing rational players with faulty players has already been extensively explored
in the context of secret sharing and multi-party computation~\cite{lysyanskaya2006rationality,FKN10,DMRS11,ADGH}. 
In particular, the central third-party mediator that is typically relied upon was implemented with synchronous \emph{cheap talks}~\cite{ADGH}, that are communications of negligible cost through private pairwise channels. This extension was indeed illustrated with an $\epsilon$-$(k,t)$-robust secret sharing protocol where $n>k+2t$.
It was later shown~\cite{abraham2019implementing} that mediators could 
be implemented with asynchronous cheap talks in an $\epsilon$-$(k,t)$-robust
protocol when $n>3(k+t)$. 
%but none of these models is suited to solve consensus 
%with partial synchrony.
%
This adaptation makes it impossible to devise even a $1$-immune 
protocol that would solve the consensus problem~\cite{LSP82} as the 
communication model becomes asynchronous~\cite{FLP85}.
In this paper, we focus instead on the partially synchronous model, 
where the bound on the delay of messages is unknown~\cite{DLS88}, 
to design a protocol that solves consensus among $n$ players, 
where up to $t$ are Byzantine players and $k$ are rational players.

Consensus has been explored in the context of game theory.
Some works focused on the conditions under which termination and validity
is obtained for a non-negligible cost of communication and/or local
computation~\cite{Vilaca2012,Amoussou-guenou2020}, without considering
the incentives for rational players to cause a disagreement. 
This incentive is quite apparent in the blockchain context, where Bitcoin users 
hacked the network to double spend by simply leading sets of players to a disagreement (or fork) for long enough\footnote{\url{https://www.cnet.com/news/hacker-swipes-83000-from-bitcoin-mining-pools/}.}.
Some results consider the problem of consensus in the presence of rational players but do not consider failures~\cite{GKT12}.
Leader election~\cite{ADH19}, which can be used to solve consensus indirectly, and 
consensus proposals~\cite{HV20} focus on ensuring fairness defined as all players
having an equal probability of their proposal being decided.
%Other works also extend consensus or leader election to an additional property, fairness~\cite{ADH19,HV20},
%defined as all players having equal probability of their proposal being
%decided.
%
Some proposals study consensus and mix faulty players with rational players~\cite{AGL14,BCH21}, however, they consider the synchronous communication model.

Several research results focus
%There are also works focusing 
more particularly on agreement, with some deriving from
the BAR (Byzantine-Altruistic-Rational) model. However, these works
considered either no Byzantine players~\cite{Groce2012,ebrahimi2019getting},
no coalitions of rational players~\cite{aiyer2005bar}, synchrony~\cite{Groce2012,harel2020consensus,vilacca2011n,HarzGGK19} or
solution preference~\cite{harel2020consensus}.
By assuming a larger payoff for agreeing than for disagreeing, 
solution preference requires that rational players never have an incentive to sabotage agreement.
To the best of our knowledge, we present the first work that considers
bounds for the robustness of agreement against coalitions of
Byzantine and rational players in partial synchrony.

The baiting strategy that we introduced to reward traitors of a coalition 
is very similar to the betrayal used in
%has been previously studied by 
%Dong et al.~\cite{dong2017betrayal} to 
the context of verifiable cloud-computing for counter-collusion contracts, assuming that the third party hosting the contracts is trusted~\cite{dong2017betrayal}. 
Our BFTCR protocol presents a novel implementation to select the winner of the baiting reward without a trusted third party.

There are a number of advantages of BFTCR compared to other state-of-the-art
protocols. One may think that a solution similar to submarine
commitments~\cite{breidenbach2018enter} would work as well by, for
example, hiding the proofs-of-fraud in a decision. However, such a solution does
not prevent Byzantine players from always hiding in a submarine commitment
their proofs-of-fraud, and revealing them only if a rational player reveals their
submarine commitment, which can act as a deterrent for
rational players to not betray the coalition. 
%not obtaining robustness to
%implement agreement. 
Additionally, other protocols based on
zero-knowledge proofs~\cite{kosba2016hawk} explicitly reveal the
existence of an information to prove, which gives an additional
advantage to other players in the coalition to also claim the same
knowledge.

State-of-the-art verifiable secret-sharing
(VSS) schemes~\cite{SPURT,HAVSS,randshare,abraham2021} are not a good
fit either, since there are no secret-sharing schemes in partial
synchrony that tolerate coalitions of size greater than a third
of the participants. To the best of our knowledge, BFTCR is the first
protocol that implements baiting strategies for consensus tolerating
coalitions of up to $k$ rational and $t$ Byzantine players as long as
$n>\max\left(\frac{3}{2}k+3t,2(k+t)\right)$.

\subsection{Roadmap}
  The rest of the paper is structured as follows:
Section~\ref{sec:model} presents our model and preliminary definitions,
Section~\ref{sec:imp} introduces the definition of a baiting strategy
and shows that it is impossible to solve the \problem problem without
a baiting strategy, in Section~\ref{sec:TRAP} we present the
\TRAP protocol and its correctness, and we finally conclude
in Section~\ref{sec-9}.

  % , for which we show in Section~\ref{sec:psync} bounds for a deposit and a reward to implement an effective baiting strategy, and we prove that the \TRAP protocol solves the \problem problem without synchrony and without solution preference. %that is $\epsilon$-$(k,t)$-robust for $n>\max(\frac{3}{2}k+3t,2(k+t))$. 
\section{Preliminaries}
\label{sec:model}
%\TODO{make sure not redundant}
%\TODO{define deposit assumption}\vincent{We assume that players can check that others have deposited. this is a bit protocol-specific.}
We consider a partially synchronous communication network, in which
messages can be delayed by up to a bound 
%$\Lambda$ 
that is unknown.
%, but this bound is
%not known.  
For this purpose, 
we adapt the synchronous and asynchronous models of Abraham et
al.~\cite{ADGH,abraham2019implementing} to partial synchrony.
%rely on the synchronous and
%asynchronous model of Abraham et
%al.~\cite{ADGH,abraham2019implementing}, with some changes to adapt
%these models to partial synchrony. 
%We first introduce the notion of
%the \textit{\scheduler}, %a.k.a. 
%which is part of 
%the environment
%%, which 
%and
%%represents the delays on the
%%communication network.  
%selects each message delay.
We consider a game played by a set $N$ of $|N|=n$ players, each of type in $\mathcal{T}=\{\text{Byzantine, rational, correct}\}$. The game is in \emph{extensive form}, described by a game tree whose
leaves are labeled by the utilities $u_i$ of each player $i$. We introduce the \scheduler as an additional player that will model the delay on messages derived from partial synchrony.
%The set of players playing the game is $N$ of
%size $|N|=n$. 
 We assume that
players alternate making moves with the \emph{\scheduler}: first the
\scheduler moves, then a player moves, then the \scheduler moves and
so on. The scheduler's move consists of choosing a player $i$ to move
next and a set of messages in transit to $i$ that will be delivered
just before $i$ moves (so that $i$'s move can depend on all the
messages $i$ delivers). Every non-leaf node is associated with either
a player or the \scheduler. The \scheduler is bound to two
constraints. First, the \scheduler can choose to delay any message $\ms{msg}$ up
to a bound, known only to the \scheduler, before which he must have
chosen all recipients of $\ms{msg}$ to move and provided them with
this message, so that they deliver it before making a move. Second,
the \scheduler must eventually choose all players that are still
playing. That is, if player $i$ is playing at time $e$, then $i$ is chosen to play at time $e'\geq e$.

Each player $i$ has some \textit{local state} at each node, which
translates into the initial information known by $i$, the messages $i$
sent and received at the time that $i$ moves, and the moves that $i$
has made. The nodes where a player $i$ moves are further partitioned
into \textit{information sets}, which are sets of nodes
in the game tree that contain the same local state for the same
player $i$, in that $i$ cannot distinguish them. We assume that
the \scheduler has complete information, so that the \scheduler{}’s
information sets consist of the singletons.

Since we do not assume synchrony, we need our game to be able to
continue even if a faulty player decides not to
reply. As such, w.l.o.g. we assume that players that decide not to play will at
least play the \textit{default-move}, which consists of notifying 
the \scheduler that this player will not move, so that the game
continues with the \scheduler choosing the next player to move. Thus,
in every node where the scheduler is to play a move, the \scheduler
can play any move that combines a player and a subset of messages that
such player can deliver before playing.  Then, the selected player
moves, after which the \scheduler selects again the next player for
the next node, and the messages it receives, and so on. The \scheduler
alternates thus with one player at each node down a path in the game
tree until reaching a leaf. A \textit{run} of the game is then a downward
path in the tree from the root to a leaf.
%We consider the game  in \emph{extensive form}, that is,
%represented as a game tree in which, starting with the \scheduler in
%the root node, the \scheduler chooses the next player to make a move in the next node down the path. The scheduler also chooses the messages that such
%player receives in that node, taken from the set of messages sent to
%that player before. 
%\vincent{What is a round?}
%After the player finishes its round, again the
%\scheduler chooses the next player for the next node, and the messages it
%receives. 
%\vincent{what happens if there is no message to receive? who sends the messages?}\ARP{then that player does not receive any messages. Other players send messages}
% The leaves of the tree are labelled by the utilities $u_i$ of each
% player $i$. 
%Each player has a type taken from the type space
%$\mathcal{T}=\{\text{Byzantine, rational, correct}\}$.

\mypar{Strategies.} We denote the set of actions of a player $i$ (or the \scheduler) as
$A_i$ (or $A_s$), and a strategy $\sigma_i$ for that set of actions is denoted as
a function from $i$'s information sets to a distribution over the
actions. We denote the set of all possible strategies of player $i$ as
$\mathcal{S}_i$. Let $\mathcal{S}_I=\Pi _{i\in I} \mathcal{S}_i$ and
$A_I=\Pi_{i\in I} A_i$ for a subset $I\subseteq N$. Let
$\mathcal{S}=\mathcal{S}_N$ with $A_{-I}=\Pi _{i\not\in I} A_i$ and
$\mathcal{S}_{-I}=\Pi _{i\not\in I} \mathcal{S}_i$.  A \textit{joint
strategy} $\vv{\sigma}=(\sigma_0,\sigma_1,...,\sigma_{n-1})$ draws
thus a distribution over paths in the game tree (given the scheduler's strategy $\sigma_s$), where $u_i(\vv{\sigma},\sigma_s)$ is player's $i$ expected utility if $\vv{\sigma}$ is played along with a strategy for the scheduler $\sigma_s$. A strategy $\theta_i$ \textit{strictly dominates} $\tau_i$ for $i$ if for all $\vv{\phi}_{-i}\in \mathcal{S}_{-i}$ and all strategies $\sigma_s$ of the \scheduler we have $u_i(\theta_i,\vv{\phi}_{-i},\sigma_s)>u_i(\tau_i,\vv{\phi}_{-i},\sigma_s)$.

Given some desired functionality
$\mathcal{F}$, a \textit{protocol} is the recommended joint strategy
$\vv{\sigma}$ whose outcome satisfies $\mathcal{F}$ for all strategies $\sigma_s$ of the \scheduler, and an \textit{associated
game} $\Gamma$ for that protocol is defined as all possible deviations
from the protocol~\cite{ADGH}. In this case, we say that the protocol
$\vv{\sigma}$ \textit{implements} the functionality. Note that
both the \scheduler and the players can use probabilistic strategies.

\mypar{Failure model.}
We set $t_0= \ceil{\frac{n}{3}}-1$ for the rest of this paper and
$k$ players out of $n$ can be rational while $t\leq t_0$ can
be Byzantine; the rest of the players are correct. \emph{Correct players} follow the
protocol: the expected utility of correct player $i$ is greater than
$0$ for any run in which the outcome satisfies consensus, and $0$ for
any other run. \emph{Rational players} can deviate to follow the strategy
that yields them the greatest expected utility at any time they are to
move, while \emph{Byzantine players} can deviate in any way, even not
replying at all (apart from notifying the \scheduler that they will
not move). Rational players have greater utility for outcomes in which
they caused a disagreement than from outcomes that satisfy consensus,
but have no interest in deviating from consensus for anything else, in
that they prefer to terminate and to guarantee validity. We will
detail further the utilities of rational players in
Section~\ref{sec:psync}.

We assume that if a coalition manages to cause a disagreement, then it
obtains a payoff of at most $\mathcal{G}$, which we call the
\textit{total gain}. Nevertheless, this total gain may be, for example,
the entire market value of the system. In a payment system application
in which players agree on a set of transactions to be decided, the total gain
$\mathcal{G}$ is exactly the sum of all the amounts spent in all transactions. We also
assume, w.l.o.g., that a coalition with $k$ rational players and $t$
Byzantine players split equally the total gain into $k$ parts, which
we call the \textit{gain} $g=\mathcal{G}/k$, that is, Byzantine
players are willing to give all the total gain from causing a
disagreement to the rational players that collude (to incentivize the
deviation for these rational players). Note that a protocol that
tolerates a maximum gain $\mathcal{G}$ equally split into $k$ parts also tolerates
any gain such that the maximum share of the split is $\mathcal{G}/k$, but we assume the equal split for ease of exposition.
We speak of the \textit{disagreeing}
strategy as the strategy in which players collude to produce a
disagreement, and of a coalition \textit{disagreeing} to refer to a
coalition that plays the disagreeing strategy. A disagreement of
consensus can mean two or more disjoint groups of non-deviating
players deciding two or more separate, conflicting
decisions~\cite{singh2009zeno}. For ease of exposition,
we consider in this work only disagreements into two values. Nonetheless, if the size of the coalition is less than half the total
number of players $k+t<n/2$ (as is the case for the work that we
present) then the coalition can only cause a disagreement into two
values~\cite{singh2009zeno}, whereas greater sizes of a coalition can
cause disagreements into multiple values~\cite{ranchal2020blockchain}.

We let rational players in a coalition and Byzantine players (in or
outside the coalition) know the types of all players, so that they
 know which players are the other Byzantine players,
rational players and correct players, while the rest of the players
only know the upper bounds on the number of rational and Byzantine
players, i.e., $k$ and $t$ respectively, and their own individual type
(that is, whether they are rational, Byzantine or correct).

\mypar{Cheap talks.}
As we are in a fully distributed system, without a trusted central
entity like a mediator, we assume \textit{cheap-talks}, that is,
private pairwise communication channels. We also assume negligible
communication cost through these channels. Non-Byzantine players are
also only interested in reaching consensus,
and not in the number of messages exchanged. Similarly, we
assume the cost of performing local computations (such as validating
proposals, or verifying signatures) to be negligible.

\mypar{Cryptography.}
We require the use of standard cryptography, for which we reuse the
assumptions of Goldreich et al.~\cite{goldreichplay}: polynomially
bounded players and the enhanced trapdoor permutations.
%We further assume the
%existence of a public-key infrastructure (PKI). 
In practice, these two assumptions mean that players can sign
unforgeable messages, and that they can perform oblivious transfer.
Each player has a public key and a private key, and public keys are
common knowledge.

\mypar{Robustness.}
Given that a Nash equilibrium only protects against single-player
deviations, and our distributed system may be susceptible of a coalition
of $k$ rational and $t$ Byzantine players, it is important to consider tolerating multi-player
deviations.  
%Given that disagreements in consensus require coalitions
%of rational players, 
We thus restate Abraham's et al.~\cite{ADGH} definitions
of %$k$-resilience, 
$t$-immunity,
$\epsilon$-$(k,t)$-robustness and the most recent definition of $k$-resilient
equilibrium~\cite{abraham2019implementing}.
%A join strategy is $t$-immune if it tolerates coalitions of up to $t$
%Byzantine faults and $k$-resilient if it tolerates coalitions of
%up to $k$ rational players. \vincent{this is a bit too simplistic as a definition.}
% HERE \ARP{(1)Maybe remove $k$-resilient to make space to explain better?}
The notion of $k$-resilience is motivated in distributed computing by
the need to tolerate a coalition of $k$ rational players that can all coordinate
actions. A joint strategy is $k$-resilient if not all rational members of a coalition of size at most $k$
can gain greater utility by deviating in a coordinated way.

%\TODO{do we need resilience here?}
\begin{defn}[$k$-resilient equilibrium]
  A joint strategy $\vv{\sigma}\in \mathcal{S}$ is a \textit{$k$-resilient equilibrium} (resp. \textit{strongly k-resilient equilibrium}) if, for all $K\subseteq N$ with $|K|\leq k$, all $\vv{\tau}_{K}\in \mathcal{S}_K$, all strategies $\sigma_s$ of the \scheduler, and for some (resp. all) $i\in K$ we have $u_i(\vv{\sigma}_K, \vv{\sigma}_{-K},\sigma_s) \geq u_i(\vv{\tau}_{K}, \vv{\sigma}_{-K},\sigma_s)$.
\end{defn}

The notion of $t$-immunity is motivated in distributed algorithms by
the need to tolerate $t$ Byzantine players. An equilibrium
$\vv{\sigma}$ is $t$-immune if non-Byzantine players still prefer to
follow $\vv{\sigma}$ despite the deviations of up to $t$ Byzantine
players.

\begin{defn}[$t$-immunity]
  A joint strategy $\vv{\sigma}\in \mathcal{S}$  is t-immune if, for all $T\subseteq N$ with $|T|\leq t$, all $\vv{\tau}\in \mathcal{S}_T$, all $i\not \in T$ and all strategies of the scheduler $\sigma_s$, we have $u_i(\vv{\sigma}_{-T}, \vv{\tau}_T,\sigma_s) \geq u_i(\vv{\sigma},\sigma_s)$.
\end{defn}

% \begin{defn}[$\epsilon$-$(k,t)$-robustness]
%   A joint strategy $\vv{\sigma}\in \mathcal{S}$  is $\epsilon$-$(k,t)$-robust (resp. strongly $\epsilon$-$(k,t)$-robust) if for all $K, T\subseteq N$ such that $K\cap T = \emptyset, |K|\leq k,$ and $|T|\leq t$, for all $\vv{\tau}_T \in \mathcal{S}_T$, for all $\vv{\phi}_K\in\mathcal{S}_K$, for some (resp. all) $i\in K$ we have $u_i(\vv{\sigma}_{-T},\vv{\tau}_T)\geq u_i(\vv{\sigma}_{N-(K\cup T)}, \vv{\phi}_K, \vv{\tau}_T)-\epsilon$.
%   \end{defn

%Similarly, 
%A joint strategy is an $\epsilon$-$(k,t)$-robust equilibrium if it tolerates 
%coalitions of up to $k$ rational players and $t$ Byzantine players.
A joint strategy is an $\epsilon$-$(k,t)$-robust equilibrium if no coalition of $k$ rational players can coordinate to increase their expected utility by $\epsilon$
regardless of the arbitrary behavior of up to $t$ Byzantine players, even if the Byzantine players join their coalition.
%\vincent{if it reaches an equilibrium tolerating $k$ players coordinated their change of strategies and if $k$ players do not change their strategy despite $t$ players deviating potentially irrationally?}\ARP{not sure I understand}. 
We
illustrate it however with $\epsilon$ because of the use of
cryptography, that is, in order to account for the (negligible) probability
of the coalition breaking cryptography, as was done previously~\cite{ADGH}:

  \begin{defn}[$\epsilon$-$(k,t)$-robust equilibrium]
  A joint strategy $\vv{\sigma}\in \mathcal{S}$  is an $\epsilon$-$(k,t)$-robust (resp. strongly $\epsilon$-$(k,t)$-robust) equilibrium
  if for all $K, T\subseteq N$ such that $K\cap T = \emptyset, |K|\leq k,$ and $|T|\leq t$, for all $\vv{\tau}_T \in \mathcal{S}_T$, for all $\vv{\phi}_K\in\mathcal{S}_K$, for some (resp. all) % \vincent{What does "all (resp. all)" mean?}
  $i\in K$, and all strategies of the \scheduler $\sigma_s$, we have $u_i(\vv{\sigma}_{-T},\vv{\tau}_T,\sigma_s)\geq u_i(\vv{\sigma}_{N-(K\cup T)}, \vv{\phi}_K, \vv{\tau}_T,\sigma_s)-\epsilon$. We speak instead of a $(k,t)$-robust equilibrium if $\epsilon=0$.
\end{defn}
% HERE\ARP{w.r.t (1)}
We use a recent definition of $k$-resilient
equilibrium~\cite{abraham2019implementing}, which slightly
differs from the definition of an $\epsilon$-$(k,t)$-robust equilibrium. We define here
strong resilience and strong robustness to refer to the stronger
versions of these properties~\cite{ADGH}. Byzantine fault tolerance in distributed computing is equivalent to our definition of $t$-immunity in game theory.

%  Previous definitions of robustness, resilience and immunity included
% the definitions for all type profiles~\cite{abraham2019implementing},
% which we omit given that our type profile is strictly defined by the
% values $k$ and $t$. That is, the type of players is either correct,
% rational or Byzantine but all rational players share the same utilities.

% HERE \ARP{w.r.t. (1)}
Given some game $\Gamma$ and desired functionality $\mathcal{F}$, we
say that a protocol $\vv{\sigma}$ is a $k$-resilient protocol for
$\mathcal{F}$ if $\vv{\sigma}$ implements $\mathcal{F}$ and is a $k$-resilient equilibrium.  % A mechanism is
% \textit{practical} if it survives iterated deletion of weakly
% dominated strategies, as defined in previous
% work~\cite{rabin1989verifiable}.
% All possibility results in this paper
% are practical mechanisms.
For example, if $\vv{\sigma}$ is a
k-resilient protocol for the consensus problem, then in all runs of
$\vv{\sigma}$, every non-deviating player terminates and agrees on the same
valid value. We extend this notation to $t$-immunity and
$\epsilon$-$(k,t)$-robustness. % We speak of $t$-immune,
% $k$-resilient and $\epsilon$-$(k,t)$-robust protocols for a
% functionality to refer to $t$-immunity, $k$-resilience and
% $(k,t)$-robustness of the associated .
The required
functionality of this paper is thus reaching agreement.
  
\mypar{Punishment strategy.}
 We also restate the definition of a punishment strategy~\cite{ADGH} as a threat that correct and rational players can play in order to prevent other rational 
 players from
deviating.
  % In the following we need to consider a set of players that will act against the coalition by bidding. %, this makes the size of $P$ a parameter.\vincent{Is this up-to-date?}
The punishment strategy guarantees that if $k$ rational players deviate, then $t+1$ players can lower the utility of these rational players by playing the punishment strategy. 
  \begin{defn}[$(k, t)$-punishment strategy]  A joint strategy $\vv{\rho}$ is a $(k,t)$-punishment strategy with respect to $\vv{\sigma}$ if for all $K,T,P\subseteq N $ such that $K,T,P$ are disjoint, $|K|\leq k,|T|\leq t,|P|> t$, for all $\vv{\tau}\in \mathcal{S}_T$, for all $\vv{\phi}_K\in \mathcal{S}_K$, for all $i\in K$, and all strategies of the \scheduler $\sigma_s$, we have $u_i(\vv{\sigma}_{-T},\vv{\tau}_T,\sigma_s) > u_i(\vv{\sigma}_{N-(K\cup T \cup P)}, \vv{\phi}_K, \vv{\tau}_T, \vv{\rho}_P,\sigma_s)$.
  \end{defn}
  Intuitively, a punishment strategy represents a threat to prevent rational players from deviating, in that if they deviate, then players in $P$ can play the punishment strategy $\vv{\rho}$ and the deviating rational players decrease their utility with respect to following $\vv{\sigma}$. For example, crime sentences are an
effective punishment strategy against committing crimes. Not
terminating a protocol if just one player deviates can also be a
punishment strategy against deviating from the protocol.

  % If the punishment strategy does not make use of cryptography, then 
  % vincent: this is not a generalization
%  we generalize the definition to a $(k,t)$-punishment strategy, whose definition is equivalent except for the absence of $\epsilon$ in it.
% we simply consider $\epsilon = 0$ and write $(k,t)$-punishment strategy.
  % \ARP{Slight variation of ktp-punishment: we do not put restrictions on the size of P be}
      % \ARP{with this definition of punishment strategy, we can define some correct processes that will play $\vv{\rho}$ being $\vv{\rho}$ 'bait a trap as if they were going to join an attack but then they expose it' and thus the rational utility of trying to deviate can decrease. This requires altruistic}

\mypar{Accountability.}
 Previous work introduced signatures in consensus protocol messages,
guaranteeing that for a disagreement to occur, at least $t_0+1$ players
must sign conflicting messages, and once these messages are discovered
by a correct player, such player can prove the fraudsters to the rest of
correct players through \textit{Proofs-of-Fraud
(PoFs)}~\cite{civit2020brief, ranchal2020blockchain}. We also adapt to this model the property of accountability, recently defined for consensus~\cite{civit2020brief,CGG21}:

      \begin{defn}[accountability]
        Let $\vv{\sigma}$ be a protocol that implements
agreement. Suppose that a disagreement takes place, then $\vv{\sigma}$
is accountable if all correct players will eventually gather
enough proof that at least $t_0+1$ players deviated to
cause the disagreement.%, for $t=\ceil{\frac{n}{3}}-1$
  \end{defn}

% In order to distinguish the value $t$
% from $t$-immunity and from $(k,t)$-robustness, we set
% $t=\ceil{\frac{n}{3}}-1$ in the rest of this paper. We will make
% use of $f$ instead to refer to more restrictive bounds for Byzantine
% players, i.e., $f\leq t$.
 
\mypar{Rational agreement.}
In the remainder, we are interested in proposing a consensus protocol that is immune to up to $t_0$ Byzantine failures and robust to a coalition of up to $k$ rational and $t$ Byzantine players, so we restate the Byzantine consensus problem~\cite{LSP82} in the presence of rational players:
The \emph{Byzantine consensus problem} is, given $n$ players, each with an initial value, to ensure (i)~\emph{agreement} in that no two non-deviating players decide different values;
(ii)~\emph{validity} in that the decided value has to be proposed; and (iii)~\emph{termination} in that eventually every non-deviating player decides.

\begin{defn}[\Problem]
  Consider a system with $n$ players, a protocol $\vv{\sigma}$ solves the \problem problem if it implements consensus, and is $t_0$-immune and $\epsilon$-$(k,t)$-robust for some $k, t>0$ such that $n \leq 3(k+t)$.
\end{defn}

% \TODO{Let $t=\ceil{n/3}-1$,  for the remainder of the paper (we will use it)}
% \ARPN{}
\section{Rational Agreement Impossibility without a Baiting Strategy}
\label{sec:imp}

In this section, we introduce a baiting strategy as a particular case
of punishment strategy and show that it is necessary to devise a
consensus protocol robust to a coalition of $k$ rational players and
$t$ Byzantine players.

% We defer all proofs to the Appendix~\ref{app:proofs}.

  Our solution to agreement in the presence of rational and Byzantine
players, presented in Section~\ref{sec:psync}, consists of rewarding
rational players for betraying the coalition.  One may wonder whether
rewarding rational players in a coalition is the only way to obtain
$\epsilon$-$(k,t)$-robustness that tolerates coalitions of size
$n\leq 3(k+t)$ in partial synchrony. To demonstrate the need for a
reward, we first formalize a type of $(k,t)$-punishment
strategy, which we call a $(k,t,m)$-baiting strategy. A
$(k,t,m)$-baiting strategy is a
$(k-m,t)$-punishment strategy such that $k\geq
m > 0$, and these $m$ rational players prefer to actually play the
baiting strategy than to deviate with the rest of the players in the
coalition. That is, $m$ players of the coalition have to play the
baiting strategy for it to succeed, and at least $m$ rational players
in the coalition prefer to play the baiting strategy than to deviate
with the coalition. An example is offering a crime reduction for a criminal to cooperate with law enforcement into catching the criminal group to which it belongs.
  
  \begin{defn}[$(k, t, m)$-baiting strategy]\label{def:ekf-bs}
    A joint strategy $\vv{\eta}$ is a $(k,t,m)$-baiting strategy with respect to a strategy $\vv{\sigma}$ if $\vv{\eta}$ is a $(k-m,t)$-punishment strategy with respect to $\vv{\sigma}$, with $0< m \leq k$ and for all $K,T,P\subseteq N$ such that $K\cap T =\emptyset ,\, |P\cap K|\geq m,\, P\cap T = \emptyset$, $|K\backslash P|\leq k-m,|T|\leq t, |P|> t$, for all $\vv{\tau}\in \mathcal{S}_T$, all $\vv{\phi}_{K\backslash P}\in \mathcal{S}_{K\backslash P}-\{\vv{\sigma}_K\}$, all $\vv{\theta}_{P}\in \mathcal{S}_{P}$, all $i\in P$, and all strategies of the \scheduler $\sigma_s$, we have:
    \begin{equation*}
      \hspace{-1.9em}
      u_i(\vv{\sigma}_{N-(K\cup T \cup P)}, \vv{\phi}_{K\backslash P},\vv{\tau}_T, \vv{\eta}_P, \sigma_s)\geq   
      u_i(\vv{\sigma}_{N-(K\cup T \cup P)}, \vv{\phi}_{K\backslash P},\vv{\tau}_T,\vv{\theta}_{P},\sigma_s).
    \end{equation*}
    % while for all $i\in K$:
    % \begin{equation*}
    %   u_i(\vv{\sigma}_{N-(K\cup T \cup P)}, \vv{\phi}_K, \vv{\tau}_T, \vv{\eta}_P)<       u_i(\vv{\sigma}_{-T}, \vv{\tau}_T)
    % \end{equation*}
    Additionally, we speak of a strong $(k,t,m)$-\textit{baiting strategy} in the particular case where for all rational coalitions $K \subseteq N$ such that $|K|\leq k$, $|K\cap P|\geq m$ and all $\vv{\phi}_{K\backslash P}\in \mathcal{S}_{K\backslash P}$ we have:
    % \begin{equation*}
    $\sum_{i\in K}u_i(\vv{\sigma}_{N-(K\cup P)}, \vv{\phi}_{K\backslash P}, \vv{\eta}_P,\sigma_s) \leq \sum_{i\in K}u_i(\vv{\sigma}, \sigma_s).$
    % \end{equation*}
    % We refer to a $(k,t,m)$-baiting strategy (resp, strong $(k,t,m)$-baiting strategy)  %and strong $(k,t,m)$-baiting strategy 
    % as an $\epsilon$-$(k,t,m)$-baiting strategy (resp. $\epsilon$-$(k,t,m)$-strong baiting strategy) where $\epsilon=0$.
    \end{defn}

    A baiting strategy illustrates a situation where at least $m$ rational players in the coalition may be interested in baiting other $k+t-m$ rational and Byzantine players into a trap: the $k+t$ of them collude to deviate initially, just so that these $m$ players can prove such deviation by playing the baiting strategy, and get a reward for exposing this deviation. Such a strategy has a significant impact in a protocol to implement agreement. A strong baiting strategy defines a baiting strategy in which the fact that $m$ deviating players play the baiting strategy does not yield greater payoff to the entire coalition as a whole (if such coalition was made only by rationals), compared to following the protocol. This prevents a coalition of rational players from colluding together so as to play the baiting strategy on themselves only with the purpose of splitting the baiting reward among the colluding members. Notwithstanding, neither a baiting strategy nor a strong baiting strategy show that if these $m$ players play the baiting strategy, then the protocol implements the desired functionality. We illustrate the efficacy of baiting strategies to influence the outcome of a protocol in the example of the rational generals, shown in Figure~\ref{ex:ratgen}.
\begin{figure}[ht]
\noindent\fbox{%
    \parbox{\columnwidth-1em}{ \textbf{Rational generals example.} We
illustrate the intuition behind baiting strategies with an example
inspired from the Byzantine generals problem~\cite{LSP82}, that we refer to as the
`rational generals' problem: suppose $n=7$ Ottoman generals need to
agree on whether to attack or retreat. If all generals agree on
attacking, they will succeed, if they agree on retreat, they can
succeed another day. However, if only some of the generals attack,
they will lose. There are two Byzantine generals, i.e., $t=2$ , whose
goal is for the Ottomans to disagree on their decision for them to
lose, and another rational general, i.e., $k=1$, who has been offered a
bribe $\mathcal{G}$ in order to contribute to the disagreement, but
who is willing to betray the Byzantines for a greater income from
the Ottomans. Because of accountability, the generals will eventually
be able to track the disagreement to both the $t$ Byzantine and $k$ rational generals, but by then the
$k$ rational generals
will be enjoying its reward $\mathcal{G}$ in Constantinople, out of
reach.

\vspace{0.5em}\hspace{1em}
The generals suspect that there might be a bribed rational general ($k=1$). In an
attempt from them to make the rational general talk, they offer a
reward $\mathcal{R}>\mathcal{G}$ as a bounty for proving the fraud of
every other Byzantine and rational general, that is, if the rational general reveals
its identity and that of the $t$ Byzantine with proofs, then this rational general is spared and
rewarded with $\mathcal{R}$, while the $t$ Byzantine generals lose all of their capital (i.e., properties and
savings) that they own in the Ottoman empire. In this case, the rational general sees a greater incentive
%from 
to expose both himself and the Byzantine generals. This is an
example of a baiting strategy. Additionally, the Ottoman generals will pay
$\mathcal{R}$ with the capital taken from the $t$ Byzantine generals, so the Ottoman empire
will not even pay for the reward.

\vspace{0.5em}\hspace{1em}
Notice that Ottoman generals must guarantee to the rational general
that they will recognize him as the first to expose the coalition (and
the only rightful owner of the reward), so that the rational general
is not influenced by a threat from the Byzantine generals to steal the
reward if he betrays the coalition. That is, the rational general
will only bait the coalition if the protocol ensures that the
Byzantine generals will not be able to steal the reward from the
rational general after seeing that he betrayed the coalition. This is
in order to prevent the Byzantine generals from rushing to bait as soon as they learn the rational general is starting to bait, creating a situation
in which both Byzantine and rational generals seem to be legitimate
baiters of the coalition.

\vspace{0.5em}\hspace{1em}
In the extensive game, this means that the rational general must first
behave and make moves as if he would cause the disagreement. Then, the
rational general will only bait if he gets both enough evidence of the
fraud of the deviants and assurance that the Byzantine generals will
not outpace him and steal the reward.}
}
\caption{Rational generals example.}
\label{ex:ratgen}
\end{figure}

    % if the sum of the utilities of the $m$ players and the players that suffer the punishment is greater for the baiting strategy than for following the protocol, then a rational coalition might be interested in baiting itself, just to get the baiting reward, which then the coalition can split between all members of the coalition, or force a lottery to choose a winner of the reward\vincent{Would this be negative in the sense that all rationals would then stop following the consensus protocol?}
    % . Therefore, we define a strong baiting strategy as a baiting strategy where baiting oneself is not profitable for rational players anymore\vincent{but baiting should actually be profitable for some rational to bait, no?}, even if controlled by the same user, forced to correlate their strategies into splitting funds, or running a lottery.

\mypar{Impossibility result.}
The reason why a $(k,t,m)$-baiting strategy is relevant to the consensus problem is that without such a strategy it is not possible to obtain a consensus protocol that is $(k,t)$-robust where $k>0$.
% consensus problem in partial synchrony stems from the fact that 
%  The need for an $\epsilon$-$(k,t,m)$-baiting strategy for the consensus problem in partial synchrony is immediate\vincent{does it mean one of our contributions is trivial?}: 
%  it is not possible to obtain any resilience if the protocol must also be $t$-immune (i.e., a BFT protocol) without an $\epsilon$-$(k,t,m)$-baiting strategy that some players can play. 
We show this result in Theorem~\ref{thm:imp}. The proof is similar to that of the impossibility of $t$-immune consensus under partial synchrony for $t>t_0$~\cite{DLS88}, since a partition of rational and Byzantine players can exploit two disjoint partitions of correct players to lead them to different decisions. Let us recall that we do not assume solution preference, and thus the payoffs from a disagreement can be singificantly greater than those of agreeing for rational players. For the proof of Theorem~\ref{thm:imp}, we first show the more general proof of Lemma~\ref{lem:imp}.
  \begin{lem}
    % Let a game $\Gamma$ that implements agreement without mediator in partial synchrony through the mechanism $(\Gamma, \vv{\sigma})$. Then there must be a $(k,t)$-baiting strategy, $p>0$, $\vv{\rho}$ with respect to $\vv{\sigma}$ in order to obtain a practical mechanism $(\Gamma, \vv{\sigma})$ that implements the agreement functionality, is $t$-immune and $(k,t)$-robust, $k\geq 0$ and $t=\max(t-k+1,0)$.
    It is impossible to obtain a protocol $\vv{\sigma}$ that implements agreement, is $t_0$-immune and $(k,t)$-robust, $k\geq 0$ and $t=\max(t_0-k+1,0)$ unless there is a $(k,t,m)$-baiting strategy with respect to $\vv{\sigma}$, for $m>\frac{k+t-n}{2}+t_0$.
    \label{lem:imp}
  \end{lem}
  \begin{proof}
    We refer to Dwork et al.'s~\cite{DLS88} work for the impossibility of increasing $t> t_0$ and obtaining agreement (i.e., for $k=0$). For $k>0$ with $t\leq t_0$, assume the contrary: let $\vv{\sigma}$ be a protocol such that there is no $(k,t,m)$-baiting strategy with respect to $\vv{\sigma}$ and $\vv{\sigma}$ is $(k,t)$-robust, for $t=max(t_0-k+1,0),\, k>0$. Since the protocol is $t_0$-immune and it works under partial synchrony, the protocol must not require more than $n-t_0$ players participating in it in order to take a decision, or else the Byzantine players could prevent termination. Consider a partition of the network between 4 disjoint subsets $N=K\cup A \cup B \cup F$, where $K$ are the rational players (there is at least one), $F$ are the Byzantine players, i.e., $|F|+|K|=t+k\geq t_0+1$, and $A$ and $B$ are the rest of the players such that $|A|+|B|\leq n-t_0-1$ and both $|A|+|F|+|K|\geq n-t_0$ and $|B|+|F|+|K|\geq n-t_0$ hold (recall $t_0=\ceil{\frac{n}{3}}-1$). Let $\vv{\theta}$ be the strategy in which the rational players in $K$ deviate with Byzantine players in $F$ and achieve a disagreement between players in $A$ and players in $B$. If the players in $F$ and $K$ are all Byzantine and rational players, then
such a disagreement is always possible and the utility for each
rational player is, by definition of the model, greater than that of reaching
agreement. Notice also that since $t=\max(t_0-k+1,0)$, if $m>\frac{k+t-n}{2}+t_0$ rational players do not deviate to cause such disagreement, we have that at least one of $|A|+|F|+|K|-m< n-t_0$ and $|B|+|F|+|K|-m< n-t_0$ holds, or both: for this value of $m$ the deviants cannot cause a
disagreement. However, this is not true if instead $m\leq \frac{k+t-n}{2}+t_0$. It follows that it is necessary to encourage at least
$m>\frac{k+t-n}{2}+t_0$ rational players to not deviate into causing a disagreement, which means, by definition, that a $(k,t,m)$-baiting strategy is necessary.
\end{proof}
\begin{thm}
  % \label{thm:imp}
  It is impossible to obtain a protocol $\vv{\sigma}$ that implements rational agreement unless there is a $(k,t,m)$-baiting strategy with respect to $\vv{\sigma}$, for $m>\frac{k+t-n}{2}+t_0$.
  \label{thm:imp}
\end{thm}
  \begin{proof}
    By definition, every $(k,t)$-robust protocol for $n \leq 3(k+t)$ must also be $(k,t)$-robust, for some $k\geq 0$ and $t=\max(t_0-k+1,0)$. Therefore it derives from Lemma~\ref{lem:imp}.
  \end{proof}
  % \begin{proof}
  %   By definition, every $(k,t)$-robust protocol for $n \leq 3(k+t)$ must also be $(k,t)$-robust, for some $k\geq 0$ and $t=\max(t_0-k+1,0)$. Therefore it derives from Lemma~\ref{lem:imp}.
  % \end{proof}

  Theorem~\ref{thm:imp} shows the need for a baiting strategy to solve
\problem. In Section~\ref{sec:bftcr} we show the
implementation of an additional phase to an accountable consensus
protocol in order to provide the functionality of a baiting
strategy. In Section~\ref{sec:psync} we illustrate the values of a
reward and deposit per player to make a strong baiting strategy that
at least $m$ rational players will play.

%   In Section~\ref{sec:psync} we illustrate
% the values of a reward and deposit per player to make a strong baiting
% strategy that at least $m$ rational players will play. We show the implementation for such a baiting strategy in
% Section~\ref{sec:proc}. % we solve \problem with cheap-talk to provide
% % bounds for robustness in a fully distributed protocol.

% It is interesting to note that a baiting strategy is necessary to solve the rational agreement problem in the 
% partially synchronous model. The result is deferred to Theorem~\ref{thm:imp} in Appendix~\ref{app:baiting-necessity}.
\section{\TRAP: Reaching Rational Agreement}
\label{sec:TRAP}
In this section, we present the \TRAP (\textbf{T}ackling \textbf{R}ational \textbf{A}greement through \textbf{P}ersuasion) protocol, the first protocol to solve the \problem problem. The \TRAP protocol comprises three components:
\begin{enumerate}
\item
% (i)
A financial component, consisting of a deposit per player $\mathcal{L}$, taken at the start of the protocol from each participating player, and a reward $\mathcal{R}$, which is given to a player in the event that it provides PoFs for a disagreement on predecisions.
\item
% (ii)
An accountable consensus component, that pre-decides outputs from an accountable consensus protocol. 
\item
% (iii)
A baiting component, embodied in a novel Byzantine Fault Tolerant \textit{commit-reveal} (BFTCR) protocol that executes after the accountable consensus protocol. This component terminates either deciding one output (predecision) of the accountable consensus protocol, or resolving a disagreement on predecisions by rewarding one of the deviating players that exposed the disagreement and punishing the rest of deviating players.
\end{enumerate}
We first provide an overview of the properties that we aim at for the
\TRAP protocol in Section~\ref{sec:proc}, and the possible runs of
the game that derive from implementing a strong baiting strategy for
the \problem problem with the aforementioned components. We then
introduce and prove the correctness of the baiting component, the BFTCR protocol, in
Section~\ref{sec:bftcr}. Finally, we analyze the financial component,
that is, the specific values of reward and deposits, in
Section~\ref{sec:psync}. The accountable consensus component can be
any accountable consensus protocol~\cite{civit2020brief,CGG21,SWN21,CGG22}, and thus we treat this
component as a black box, for the sake of generality.
\subsection{Overview: consensus with a baiting strategy}
\label{sec:proc}
% \ARP{``How many $m$ to influence BFTCR to disagreement?''}
%\section{Solving \problem}
%\TODO{change section title even more maybe?}
  
%  In this section, we show how to implement a strong baiting strategy by introducing a deposit and a reward.
% %we introduce the \problem problem as the problem of devising a 
% %consensus protocol that is resilient to a coalition of $k$ rational players and 
% %immune to \vincent{up to?} $t$ Byzantine players.
% Note that 
We proved in Section~\ref{sec:imp} that we need a baiting strategy for a protocol to solve the \problem problem.
% solve consensus and \vincent{did we prove that as well:}be $(k,t)$-robust for $n \leq 3(k+t)$, i.e., to solve the \problem problem.

Before we present the implementation of such a baiting strategy in
Section~\ref{sec:bftcr}, with additional configurations of the
required deposits and reward sizes in Section~\ref{sec:psync}, we
present in this section the basics of our baiting strategy. For this
purpose, we focus first on the properties that we aim at for such a
baiting strategy. Then, we showcase all the possible runs of a
protocol for consensus that provides such a strong baiting
strategy.

Given a protocol $\vv{\sigma}$ that implements accountable
consensus and is $t_0$-immune, we will extend it to implement the
\problem problem, in that we will prove the three following
properties:

\begin{itemize}[wide, labelwidth=!, labelindent=0pt]% [leftmargin=* ,wide=\parindent]
% \item \emph{\expandafter\MakeUppercase\rdominance}: There is a $(k,t)$-baiting $\vv{\eta}$ strategy with respect to $\vv{\sigma}$
% \item \emph{\expandafter\MakeUppercase\baitingagreement}: If at least $m$ players, with $0< m \leq k$, out of the $k$ deviating rational players play $\vv{\eta}$, then the protocol implements agreement
\item \emph{\expandafter\MakeUppercase\rdominance}: There is a $(k,t,m)$-baiting strategy $\vv{\eta}$ with respect to $\vv{\sigma}$, for $m>\frac{k+t-n}{2}+t_0$.
\item \emph{\expandafter\MakeUppercase\baitingagreement}: $\vv{\eta}$ implements agreement. % If at least $m$ players, with $0< m \leq k$ out of the $k$ deviating rational players play $\vv{\eta}$, then the protocol implements agreement.
\item \emph{\expandafter\MakeUppercase\zeroloss}: $\vv{\eta}$ is a strong baiting strategy.
\end{itemize}

\expandafter\MakeUppercase\rdominance states the necessary condition
that a baiting strategy exists, while \baitingagreement guarantees
that playing such a baiting strategy still leads to
agreement. \expandafter\MakeUppercase\zeroloss guarantees that such a
baiting strategy is a strong baiting strategy. Coming back to the
rational generals example of Figure~\ref{ex:ratgen}, \rdominance states the existence of the
reward for the rational general, \baitingagreement guarantees that
generals will still decide whether to attack or retreat after paying
the reward to the rational general, and \zeroloss guarantees that only
the slashed capital of the Byzantine generals will be used to pay the
reward to the rational general.

\mypar{Reward for baiting.} Since the protocol is
accountable, we add a \textit{baiting reward} $\mathcal{R}$ for
player $i$ if $i$ can prove to the rest of the players that a coalition of at least $t_0+1$ players are trying to cause a disagreement, but before they succeed at causing the disagreement. If multiple players are eligible for the baiting reward, then only one is chosen at random to win the reward, and the rest are treated as fraudsters that did not bait. We select the winner at random in an additional \textit{winner consensus} in which the winner is decided from among the proposed candidates to win from correct replicas in this winner consensus. We explain further the winner consensus later in this section. Players can prove that a coalition is trying to cause a disagreement through PoFs which undeniably show two conflicting messages signed by the same set of players. The reward is only given to $i$ if $i$ exposes this coalition before the coalition causes the disagreement (i.e., before both partitions of correct players decide different decisions). % In order to reduce the expected utility from playing the baiting strategy,

\mypar{Funding the reward with deposits.} we require all players to
place a minimum \textit{deposit} $\mathcal{L}$. We also require such
deposit to be big enough so that the deposit taken from the exposed
coalition is enough to pay the reward, satisfying \zeroloss. Our goal is to set $\mathcal{R}$ and $\mathcal{L}$ so that we
implement a baiting strategy for a set $M$ of rational players in the
coalition, such that if others in the coalition bait, then for all
$i\in M$, player $i$ is better off also trying to bait and getting the
reward, while if the rest of the players in the coalition do not bait,
then if $i$ baits then $i$ gets the greatest expected utility that it
can in that information set. We analyze in Theorem~\ref{thm:usefulm}
the required values for such deposit and reward necessary to
incentivize at least $|M|=m>\frac{k+t-n}{2}+t_0$ rational players in a
coalition to follow a baiting strategy, depending on the size $k+t$ of
the coalition and on the maximum total gain from disagreeing
$\mathcal{G}$. For now, however, let us ignore the values of
$\mathcal{L}$ and $\mathcal{R}$ and focus on the protocol that solves
the \problem problem, by assuming that these values of $\mathcal{L}$
and $\mathcal{R}$ are enough to make $m>\frac{k+t-n}{2}+t_0$ rational
players bait the coalition, instead of terminating a disagreement. We will come back to specify proper values for
$\mathcal{L}$ and $\mathcal{R}$ in Section~\ref{sec:psync}.  If these
PoFs expose at least $t_0+1$ players including the winner of the
baiting reward $\mathcal{R}$, then the $t_0$ (or more) remaining
colluding players lose the deposit amount $\mathcal{L}$.

\mypar{Dominating disagreements.} We explore here the possible runs, assuming that we already have such
a baiting strategy, and what each of these runs means for the payoffs
of a rational player $i$:
    \begin{enumerate}[leftmargin=* ,wide=\parindent]
      \item \label{str:cor}Rational players including $i$ contribute to reaching agreement and follow
the protocol $\vv{\sigma}$, getting some utility
$u_i(\vv{\sigma}_{-T},\vv{\tau}_T)\geq \epsilon$ where $\epsilon>0$.
      \item \label{str:dis}Some rational players collude with $i$ and deviate to disagree, playing
strategy $\vv{\phi}$ with some Byzantine players $T$ and other rational players $K$
such that $|K\cup T|\geq n/3$, $K\cap T=\emptyset$, obtaining utility $u_i(\vv{\sigma}_{N-K-T},\vv{\phi}_{K\cup T})=g$.
% \leq
% \frac{\mathcal{G}}{|K|}

      \item \label{str:bai} Player $i$ deviates to bait other rational players
into colluding with some Byzantine players such that $|K\cup T|\geq n/3$,
$K\cap T=\emptyset$, and this deviation consists of playing strategy $\vv{\eta}$ to expose the colluding players via
PoFs and obtain the baiting reward. As a result, player $i$ obtains utility $u_i(\vv{\sigma}_{N-K-T},\vv{\phi}_{K\cup T - M},\vv{\eta}_{M})=
p(m)\mathcal{R}-q(m)\mathcal{L}$, where $M$ is the set of players of the coalition that bait, i.e., $i\in M$, with $|M|=m$. $p(m)=1/m$ represents the probability of winning the reward, while $q(m)=1-p(m)=(m-1)/m$ the probability of not winning it after baiting.%=\frac{\mathcal{V}}{n/3-t}+\epsilon,\epsilon>0
      \item \label{str:suf} Player $i$ deviates to disagree only to
suffer a trap baited by another rational (or group of rational players), obtaining utility
$u_i(\vv{\sigma}_{N-K-T},\vv{\phi}_{K\cup T
  -M},\vv{\eta}_{M})\leq -\mathcal{L}$.
      \item \label{str:ben}In any run where the protocol does not terminate, player $i$ obtains negative utility.

      \item \label{str:ter}Player $i$ contributes to reaching agreement but a coalition
        causes a disagreement. In this case, $i$ is one of the victims of a disagreement (for example, a double-spending). Hence, $i$ obtains negative utility. 
\end{enumerate}

% Notice that we ignore the strategy where rational players do not validate
% proposals: we have assumed the cost of validating is negligible
% (i.e.,  free)\vincent{Can't we precisely define what costs in the model, to avoid making this kind of statements later?}, and the risk of deciding on an invalid proposal is enough
% of a threat for rational players to validate proposals.
Notice that the runs~\ref{str:suf},~\ref{str:ben} and~\ref{str:ter} are strictly dominated by run~\ref{str:cor} (following the protocol). Our goal is to make runs represented by~\ref{str:bai} runs that also implement agreement and that strictly dominate runs represented by~\ref{str:dis}.

\subsection{Baiting component: the BFTCR protocol}
\label{sec:bftcr}
% In partial synchrony, it is trivial to obtain a $k$-resilient protocol that implements agreement for $k<n$ and without Byzantine players, i.e., $t=0$: players do not decide until they hear from all other players with the same decision.
% \vincent{does not look that trivial to me: do you mean a $k$-resilient protocol that solves the rational agreement problem? can some players decide twice?}\ARP{discuss, I think it's trivial. If not, just remove this sentence}. 
% Similarly, it is possible to obtain $t_0$-immune protocols that implement agreement~\cite{DLS88}. % Before deepening into better results, we introduce the deposit per player $\mathcal{L}$ and a baiting reward $\mathcal{R}$. For this purpose, we assume a $t$-accountable consensus protocol, such as Polygraph~\cite{civit2020brief}.
In this section, we present the first implementation of a baiting
strategy for \problem. As such, we extend an accountable consensus
protocol with a Byzantine Fault Tolerant \textit{commit-reveal}
(BFTCR) phase in order to solve consensus even if there is a
disagreement at consensus level, if at least $m$ rational players
decide to betray the coalition in exchange for trying to win a reward.
We show in Algorithm~\ref{alg:phase} the BFTCR phase. As such, we speak of a
\textit{predecision} for a decision of the accountable consensus
protocol, whereas a \emph{decision} now refers to the outcome of the BFTCR
protocol. The BFTCR phase consists of 5 main parts:
\begin{enumerate}
\item
% (i)
A reliable broadcast, in which players share their encrypted commitment (line~\ref{lin:sta}),
\item
% (ii) 
a second reliable broadcast, in which players share the first $(n-t_0)$ encrypted commitments that they delivered in the first reliable broadcast (line~\ref{line:broadcast}),
\item
% (iii)
a regular broadcast, in which players share the key to reveal their commitment (line~\ref{line:bc1}),
\item
% (iv) 
an additional consensus to select the winner of the reward, if some players reveal a list of PoFs (line~\ref{line:select}), and 
\item
% (v) 
a slashing of the deposits from the fraudsters, payment of the reward to the winner and resolution of the disagreement on predecisions (line~\ref{line:punish})% -\ref{line:resolve})
. 
\end{enumerate}

\mypar{Commit and reveal.} The purpose of the first group of reliable broadcasts is to reliably broadcast the encrypted PoFs, should a player own them, or an encrypted hash of a predecision otherwise. We say that the \textit{commitment} is the encrypted content that each player decides to broadcast in this first reliable broadcast. In line~\ref{line:broadcast} each player $i$ then starts the second reliable broadcast by broadcasting a list of the first $(n-t_0)$ delivered commitments that $i$ delivered in the first reliable broadcast. The purpose of the calls to broadcast on lines~\ref{line:bc1} and~\ref{line:bc2} is to deliver the keys to decrypt the encrypted messages. A player $i$ thus \textit{reveals} his commitment by broadcasting the key. A player $i$ decrypts the commitment of player $j$ in line~\ref{line:decrypt}. Then, player $i$ adds this decrypted message to the list of decided hashes in lines~\ref{line:deliver1} to~\ref{line:deliver}, or to the list of PoFs received in lines~\ref{lin:pofr1} to~\ref{lin:pofr2}.

% \mypar{\TRAP protocol.} As such, the \TRAP protocol consists of an accountable consensus protocol (e.g., Polygraph~\cite{civit2020brief,CGG21}) enriched with a BFTCR phase execution before deciding, with a deposit per player and a reward for baiting that we detail in Section~\ref{sec:psync}.

\mypar{Termination.} The BFTCR phase of the \TRAP protocol terminates in one of two ways:
\begin{itemize}[leftmargin=* ,wide=\parindent]
  % \vspace{-2.5em}
\item 
% (i)
either there is no disagreement on predecisions, and then the protocol terminates when at least $(n-t_0)$
messages are decrypted with the same hash of the predecisions in
line~\ref{line:deliver}; 
\item
% (ii) 
or some players reveal a disagreement on predecisions
through PoFs, and then the protocol terminates when at least $t_0+1$
messages are decrypted (without counting players that are proven to be
false through a PoF) with a reward to a chosen baiter and a punishment
to the remaining players that are listed in the PoFs from
lines~\ref{lin:pofrec} to~\ref{line:resolve}.
\end{itemize}

Note that accountability does not guarantee that a baiter will gather
enough PoFs before a disagreement takes place. We prove that baiters
will gather enough PoFs before a disagreement takes place as part of the proof of
Theorem~\ref{thm:nhalf}. The idea is that $m$ rational players will wait to receive enough PoFs to be able to commit to bait, where $m$ is big enough to prevent termination
of either of the partitions of correct players.

\mypar{Valid candidates of the winner consensus.} We define a \textit{valid candidate} to win the reward as a member of a deviating coalition that committed to bait the coalition (by sending a commitment to a list of PoFs of the coalition in line~\ref{lin:sta}) independently of whether other $m$ players of the coalition also committed to bait, for $m>\frac{k+t-n}{2}+t_0$. The objective of the BFTCR protocol is to distinguish valid candidates from players who try to win the reward only after they learn that the disagreement will not succeed. A correct player $i$ considers a baiter $j$ as a valid candidate if $i$ can see $j$'s commitment to bait in at least $t_0+1$ messages from the second reliable broadcast. We refer to this $t_0+1$ messages as a \textit{proof-of-baiting} (PoB). The BFTCR protocol selects the winner of the bait among the list of valid candidates by executing an additional consensus, in the call to $\lit{select\_winner}$ in line~\ref{line:select}, in which all participating players propose the PoFs they know about and the valid candidates, along with the PoBs. We detail further this call later in this section.

Note that a rational player $i$ that commits to bait a coalition may deviate from Algorithm~\ref{alg:phase} in order to hinder other deviants from becoming valid candidates after $i$ reveals its commitment. This is because this way $i$ maximizes its chances of winning the reward (by minimizing the number of valid candidates for the reward). This is an expected deviation of a baiting rational player, which consists on waiting to deliver as many messages from the second reliable broadcast as possible from both partitions of correct players that suffered the disagreement on predecisions, and we show the correctness of this approach as part of the proof of Theorem~\ref{thm:usefulm}.

\begin{algorithm}[ht]
  \caption{BFT commit-reveal protocol for (correct) player $i$}%\ARP{mention in text that after restart replicas still allow certificates with removed deceitful replicas.}\vincent{done}} 
  \label{alg:phase}
  %\vspace{-0.2em}
  \smallsize{
    \begin{algorithmic}[1]
      \footnotesize{
        \Part{{\bf State}}{
          \State $\ms{enc\_msgs}$, list of delivered encrypted messages from the first group reliable 
           \State \T\T broadcasts, initially $\emptyset$
          \State $\ms{list\_enc\_msgs}$, list of delivered encrypted messages by other players from the 
           \State \T\T first group of reliable broadcasts, initially $\emptyset$
          \State $\ms{decrypted\_msgs}$, list of delivered decrypted messages from the first group of
           \State \T\T reliable broadcasts, initially $\emptyset$
          \State $\{\ms{RB}^1_j\}_{j=0}^n$, the first group of reliable broadcasts where $j$ is the source
          \State $\{\ms{RB}^2_j\}_{j=0}^n$, the second group of reliable broadcasts where $j$ is the source
          \State $\ms{hashes}$, a dictionary where keys are hashes and values are integers, initially it
           \State \T\T 
          does not contain any key or value %$\emptyset$ %, where $0$ is the default value for all keys
          \State $\ms{local\_hash}$, local hash of the predecided value, according to this player
          \State $\ms{POF\_received}$, boolean, initially \textbf{False}
          \State $i,\,i\_\ms{msg},\,i\_\ms{key},i\_\ms{enc\_msg}$, player's id, message, key, and encrypted message
        }\EndPart

        \vspace{-0.3em}
        \Statex \rule{0.45\textwidth}{0.4pt}
        
        \State $\ms{RB}_i^1.\lit{start(\ms{i\_enc\_msg})}$\label{lin:sta}\Comment{start first group of reliable broadcasts}
        
        \Statex
        
        \Part{\textbf{Upon RB-delivering} $\ms{enc\_msg}$ \textbf{from reliable broadcast} $\ms{RB}_j^1$}{\label{lin:rb21}% \Comment{if enough commitments from first group of reliable broadcasts}
          \State $\ms{enc\_msgs}[j]\gets \ms{enc\_msg}s$
          \SmallIf{$\lit{size(}\ms{enc\_msgs}\lit{)}\geq n-t_0$}{}
          \State $\ms{RB}_i^2\lit{.start(}\ms{enc\_msgs}\lit{)}$
        \label{line:broadcast}% \Comment{start second reliable broadcast sharing these delivered commitments}
          \EndSmallIf
        }\EndPart
        
        \Statex
        
        \Part{\textbf{Upon RB-delivering} $\ms{enc\_msgs_j}$ \textbf{from reliable broadcast} $RB^2_j$}{% \Comment{if enough set of commitments delivered from second group of reliable broadcasts}
          \State $\ms{list\_enc\_msgs}[j]\gets \ms{enc\_msgs_j}$
          \SmallIf{$\lit{size(}\ms{list\_enc\_msgs}\lit{)}\geq n-t_0$ \textbf{and} $\lit{size(}\ms{enc\_msgs}\lit{)}\geq n-t_0$}{}
          \State $\lit{broadcast(}\ms{i\_key,i}\lit{)}$\label{line:bc1}\Comment{reveal $i$'s commitment by broadcasting decryption key}
          \EndSmallIf
        }\EndPart
        
        \Statex
        
        \Part{\textbf{Upon delivering} $\ms{key}$ \textbf{from} $j$ \textbf{and} RB-delivering from $\ms{RB}_j^1$ and $\ms{RB}_j^2$  }{% \Comment{If ready to reveal $j$'s commitment}
          % \ARP{small assumption: we already got the message from msg1}
          \State $\lit{broadcast(}\ms{key,j}\lit{)}$\label{line:bc2}
          \State $\ms{decrypted\_msgs[j]}\gets \lit{decrypt(}\ms{enc\_msgs},\ms{key}\lit{)}$\label{line:decrypt}\Comment{decrypt it}
          \SmallIf{$\ms{decrypted\_msgs[j]}\lit{.type}= \lit{HASH}$}\Comment{if it is the hash of a predecision}
          \State $\ms{hash}\gets \ms{decrypted\_msgs[j]}.\lit{get\_hash()}$\label{line:deliver1}
          \State $\ms{hashes}[\ms{hash}] \,\text{+=}\, 1$\Comment{add to count}
          \SmallIf{$\ms{hashes}[\ms{hash}]\geq n-t_0$ \textbf{and} $\ms{local\_hash} = \ms{hashes}[\ms{hash}]$}{}
          \State $\lit{decide(}\ms{hash}\lit{)}}$\label{line:deliver}\Comment{if count for this hash reaches threshold, then decide it}
        \EndSmallIf
        \EndSmallIf
        \SmallElseIf{$\ms{decrypted\_msgs[j]}\lit{.type}= \lit{POFS}$}{}\Comment{if instead list of PoFs}
        \State $\ms{PoFs}\gets \ms{decrypted\_msgs[j]}.\lit{get\_PoFs()}$\label{lin:pofr1}
        \SmallIf{$\lit{verify(}\ms{PoFs}\lit{)}$}{$\ms{list\_PoFs[j]}\gets \ms{PoFs}$}\Comment{verify PoFs are valid}
        \State
        $\ms{POF\_received}\gets \textbf{True}$\label{lin:pofr2}
        \EndSmallIf
        \EndSmallElseIf
        \SmallIf{$\ms{POF\_received}$}{}\label{lin:pofrec}
        \State $\ms{msgs\_filtered}\gets \lit{keys(}\ms{decrypted\_msgs}\lit{)}\setminus \lit{keys(}\ms{PoFs}\lit{)}$% \Comment{if PoFs received, then count non-faulty decryption keys received}
        \SmallIf{$\lit{size(}\ms{msgs\_filtered}\lit{)}\geq t_0+1$}{}\Comment{winner consensus}
        \State $\ms{baiter,\, frauds,\, predec_1\, predec_2}\gets \lit{select\_winner(}\ms{list\_enc\_msgs},\ms{lPoFs}\lit{)}$\label{line:select} 
        \State $\lit{punish(}\ms{frauds}\lit{)}$\label{line:punish};
        % \State
        $\lit{reward(}\ms{baiter}\lit{)}$\label{line:reward};
        % \State
        $\lit{resolve(}\ms{predec_1},\,\ms{predec_2}\lit{)}$\label{line:resolve}% \Comment{punish fraudsters, reward winner, and resolve the disagreement on predecisions}
        \EndSmallIf
        \EndSmallIf
        \EndPart
      }
    \end{algorithmic}
  }
\end{algorithm}
\mypar{Correctness and randomness of the winner consensus.} We show in Theorem~\ref{thm:claim} that no deviating player can win the reward without being a valid candidate, i.e., no player can bait and win the reward after learning that other $m$ (or more) players baited. Additionally, note that the winner consensus solves consensus for $n>9/5(k+t)$ because at least $t_0+1$ provably fraudulent players of the coalition will not participate in it, as has already been shown~\cite{ranchal2020blockchain}% HERE\ARP{update this reference? give intuition also}
, and we consider $n>2(k+t)$. That is, at most $n'=n-(t_0+1)<2n/3$ players participate in the winner consensus. Since the maximum coalition size is $k+t<n/2$, then the remaining players of the coalition that could participate in the winner consensus are $t'=n/2-(t_0+1)<n/6$, and thus $t'<n'/3$ and the winner consensus solves consensus. Furthermore, the winner consensus only terminates once at least $n-t'$ proposals have been decided, which can be optimized through a democratic consensus protocol~\cite{dbft,civit2020brief,CGG21,ranchal2020blockchain}. Finally, after $n-t'$ proposals are decided upon, the participants execute an iteration of a random beacon that tolerates $t'<n'/3$ Byzantine faults~\cite{randshare,SPURT,HAVSS}, in order to select the winner of the baiting reward randomly from among any of the valid candidates that were in any of the decided proposals.

Following the winner consensus, in line~\ref{line:punish},
fraudsters are punished and the baiter is rewarded,
respectively. The call to $\lit{resolve(...)}$ resolves the
two disagreeing predecisions by deterministically choosing one of them
(i.e., lexicographical order) or, depending on the application, merging both.
          
% \subsubsection{

\mypar{Resolving a disagreement on consensus predecisions with BFTCR.}
It is clear that if there is no disagreement on the
predecisions, the BFTCR phase will terminate and satisfy consensus. We
consider here the output of the BFTCR phase in the case where there is
a disagreement into two predecisions. We speak of a disagreement on predecisions being \textit{finalized} if it becomes a disagreement on decisions (that is, on the output of the BFTCR phase). We will show in Theorem~\ref{thm:nhalf} that if
$m>\frac{k+t-n}{2}+t_0$ rational players commit to bait instead of
finalizing the disagreement on predecisions, then the \TRAP protocol still
satisfies consensus. For this purpose, we define
$m(k,t)=\floor{\frac{k+t-n}{2}+t_0}+1$ (i.e., the smallest natural
value that satisfies $m>\frac{k+t-n}{2}+t_0$). Then, we first show in
Lemma~\ref{thm:unavoidable} that if $m(k,t)$ rational players bait,
then the only possible outcome is to resolve a disagreement on
predecisions.

% Recall that we will detail in Section~\ref{sec:psync} which specific
% values of $\mathcal{L}$ and $\mathcal{R}$ guarantee that at least
% $m(k,t)$ rational players will bait in the event of a disagreement on
% predecisions.
\begin{lem}% [\baitingagreement]
Let $n$ players play the associated game of the \TRAP
protocol $\vv{\sigma}$, out of which $k$ can be rational and $t$ Byzantine,
with $n>2(k+t)$. Suppose a run in which a
coalition causes a disagreement on predecisions, and consider the
start of the BFTCR phase. Then, if $m(k,t)$ rational players of the
coalition commit to bait then the only possible outcome is to pay the
reward and resolve the disagreement on predecisions.
\label{thm:unavoidable}
\end{lem}
\begin{proof}
  First, we show that $m(k,t)$ deviating players committing to bait suffices to
prevent the disagreement on predecisions to be finalized in a
disagreement on decisions. This is analogous to the proof of
Lemma~\ref{lem:imp}. Then, we show that if $m(k,t)$ players commit to
bait, then the BFTCR phase safely terminates resolving predecisions,
with all correct players that start the winner consensus terminating
it and agreeing. Finally, we show that deviating players cannot get
the reward and also cause a disagreement, i.e., if one player
terminates the winner consensus then all correct players start it.

Suppose two predecisions $v_A,\,v_B$ that two partitions of players
not in the coalition $A$ and $B$ predecided, such that $A\cap
B=\emptyset$, and $|A|+|B|+k+t\leq n$. For $A$ to decide $v_A$
(resp. $B$ to decide $v_B$), players in $A$ (resp. $B$) must be able
to decide without hearing from players in $B$ (resp. $A$). Therefore,
$|A|+k+t\geq n-t_0$ and also $|B|+k+t\geq n-t_0$ to finalize the
disagreement. We consider now how many $m$ rational players out of $k$
must bait (i.e., must not contribute to finalizing the disagreement)
for a disagreement to necessarily fail. This value must be such that
$|A|+(k-m)+t<n-t_0$ and same for $B$'s partition, which solves to $m >
\frac{k+t-n}{2}+t_0$ (analogously to Lemma~\ref{lem:imp}).

Then, we recall that the BFTCR phase resolves predecisions, rewards
and punishes players if at least $t_0+1$ players have been exposed
through PoFs. Thus, every non-deviating player can ignore messages
received from a set containing at least $t_0+1$ players. All non-deviating players eventually converge to the same set of detected fraudsters~\cite{ranchal2020blockchain}, as all correct players broadcast the PoFs they hear from and update their detected fraudsters accordingly. As such,
 let $F$ represent the set of detected fraudsters, then for all
 $|F|\in[t_0+1,k+t]$ it follows that $n'/3> k+t-|F|$ for $n'=n-|F|$, and thus
the winner consensus tolerates deviations from the rest of rational
and Byzantine players not yet detected.

Finally, we show that if the reward is paid, then it is not possible
to cause a disagreement at decision level. We have shown in the
previous paragraph that all non-deviating players that execute the
winner consensus terminate agreeing. We must thus prove that if a
correct player terminates the winner consensus, then no
correct player can terminate deciding a predecision without
executing the winner consensus. Since $n'/3> k+t-|F|$, the winner
consensus terminates with the participation of just $2n'/3$ players,
of which at least $2n'/3 - (k+t-|F|)$ are correct. Since there are
$n-k-t$ correct players in total, if the winner consensus terminates
for some correct player, then there are at most $c=n-k-t-(2n'/3 -
(k+t-|F|))=(n-|F|)/3$ correct players that have neither learned about the
disagreement nor executed the winner consensus yet. Thus, for these
remaining correct players to not be able to decide without executing the
winner consensus, it is necessary that $c+t+k<n-t_0\iff t+k<n/2$.

Hence, as long as at least $m(k,t)=\floor{\frac{k+t-n}{2}+t_0}+1$
rational players play the baiting strategy, the only possible outcome
is for one of them to get the reward, and to resolve the disagreement
on predecisions.

% For this purpose, let $c$ be the
% minimum number of correct players that participate in the winner
% consensus for it to terminate. 

% The winner consensus starts
% for player $i$ when at least $t_0+1$ players have been proven
% fraudulent to $i$. Since the players that are provably fraudulent do
% not participate in the winner consensus, that means that there are
% $n'\leq n-(t_0+1)$ participants in the winner consensus. For the
% winner consensus to guarantee termination, it must guarantee that it
% requires at most $n-2t'$ players replying to terminate, where $t'$ are
% the remaining players from the coalition after removing the $t_0+1$
% proven to be fraudulent. Let us define a function $g(k,
% t)=\min(t+k-(t_0+1),t)$, then the winner consensus must be robust
% against $t'=g(k, t)$ colluding players, meaning it requires $n-g(k,t)$
% replies from different players, for a correct player to
% terminate. From this $n-g(k,t)$, at least $n-g(k,t)-(k+t)$ are replies
% from correct players. Since there are $n-k-t$ correct players, by the
% time the winner consensus terminates there can be up to
% $c=n-k-t-\big(n-g(k,t)-(k+t)\big)$ correct players that have not
% participated in the winner consensus. As such, if $c+k+t\geq n-t_0$
% then the coalition can still decide a predecision after receiving the
% reward from the winner consensus. We thus require $c+k+t<n-t_0\iff
% g(k,t)+k+t<n-t_0\iff \min(t+k-t_0,t)+k+t<n-t_0$ which means
% $n>\frac{3}{2}k+3t$ and $n>2(k+t)$.
\end{proof}
\begin{thm}[\baitingagreement]
  Let $n$ players play the associated game of the \TRAP protocol
$\vv{\sigma}$, of which $k$ can be rational and $t$ Byzantine, with $n>2(k+t)$. Suppose that
$m(k,t)=\floor{\frac{k+t-n}{2}+t_0}+1$ rational players in the coalition
play the baiting strategy committing to bait if they participate in a
disagreement on predecisions. 
  % \statement
  % Suppose $m(k,f)=\floor{\frac{k+t-n}{2}+t}+1$ rational players in the coalition play the baiting strategy if they participate in a disagreement on the predecisions.
  Then the \TRAP protocol solves \problem.
  \label{thm:nhalf}
\end{thm}
\begin{proof}
  The proof of $t_0$-immunity follows from the proof of Polygraph's
$t_0$-immunity and the fact that the additional BFTCR phase consists
of two Byzantine fault tolerant reliable broadcasts and one additional
broadcast per player, terminating each of them if $n-t_0$ players
follow the protocol.

For $\epsilon$-$(k,t)$-robustness, it is clear that if there is no
disagreement on predecisions, then rational and correct players
are more than $n-t_0$ and thus the protocol terminates and guarantees
validity and agreement. If there is instead a disagreement on
predecisions then, as long as $m(k,t)$ players commit to bait, by
Lemma~\ref{thm:unavoidable} the only outcome is to pay the reward and
resolve the disagreement.
\end{proof}    

% \begin{thm}[]
%   Let $n$ players play the associated game of the \TRAP protocol
% $\vv{\sigma}$, out of which $k$ can be rational and $t$ Byzantine, with
% $n>\frac{3}{2}k+3t$ and $n>2(k+t)$. Suppose that
% $m(k,t)=\floor{\frac{k+t-n}{2}+t_0}+1$ rational players in the coalition
% play the baiting strategy committing to bait if they participate in a
% disagreement on the predecisions. 
%   % \statement
%   % Suppose $m(k,f)=\floor{\frac{k+t-n}{2}+t}+1$ rational players in the coalition play the baiting strategy if they participate in a disagreement on the predecisions.
%   Then the \TRAP protocol solves the \problem problem.
%   \label{thm:nhalf}
% \end{thm}
% \begin{proof}
% It is clear that if there is no disagreement on the predecisions, then rational and correct players are more than $n-t_0$ and thus the protocol terminates and guarantees validity and agreement.
%   If there is instead a disagreement on the predecision then, as long as $m(k,t)$ players play the baiting strategy, by Theorem~\ref{thm:unavoidable} the only outcome is to pay the reward and resolve the disagreement.
% \end{proof}

 We show in Theorem~\ref{thm:nhalf} that, provided $m(k,t)$ rational players commit to bait if there is a disagreement in the predecisions, the \TRAP protocol solves the rational agreement problem. We only have left to prove for which values of $\mathcal{L}$ and $\mathcal{R}$ we can guarantee that the strategy to bait the coalition strictly dominates that of terminating a disagreement for at least $m(k,t)$ rational players in the coalition. We do this in Section~\ref{sec:psync}.
\subsection{Financial component: deposits \& reward} % resolving disagreements on predecisions with BFTCR
\label{sec:psync}
  \label{sec:problem}
In this section, we focus on the key idea of this paper: what are the
values required for a deposit per player and a reward to players for
baiting the coalition that make a strong baiting strategy. In particular, and
derived from the BFTCR algorithm of Section~\ref{sec:proc}, we focus
on a baiting strategy that at least $m(k,t)$ rational players will
play in Theorem~\ref{thm:usefulm}. Then, we prove that the proposed \TRAP protocol implements \problem and is $\epsilon$-$(k,t)$-robust for $n>\frac{3}{2}k+3t$ and $n>2(k+t)$ in Theorem~\ref{thm:claim} and Corollary~\ref{thm:everything}.   % Since the remaining players want to win the reward on their, they do not collude to influence the result of the winner consensus, and this result depends directly on the arbitrary local state of correct players, the winner of such consensus is indistinguishable from random before it starts. %\TODO{can make into a proof I guess...?}
We show in Theorem~\ref{thm:usefulm} which values of $\mathcal{L}$ and
$\mathcal{R}$ make the disagreeing strategy a strictly dominated
strategy by the baiting strategy for at least $m(k,t)$ rational
players (i.e., a dominated strategy even if player $i$ already knows
that $m(k,t)-1$ other players are also baiting at the time that $i$ has
to decide whether to bait or not). In other words, we show in
Theorem~\ref{thm:usefulm} under which values of $\mathcal{R}$ and
$\mathcal{L}$ such strategy $\vv{\eta}$ is a strong $(k,t,m(k,t))$-baiting
strategy that satisfies \rdominance and \zeroloss.

The result of Theorem~\ref{thm:usefulm} is the key part of the
\TRAP protocol for two reasons. First, because it shows that the
first $m-1$ baiters do not even prevent a disagremeent from taking
place, and thus if the rest of $t+k-(m-1)$ colluding players want to
finalize the disagreement, they can. Second, because it shows that if $m-1$ players commit to bait, then the remaining
$t+k-(m-1)$ must take the decision on whether to commit to bait or not
independently of what the rest of them are doing. Thus, this is
analogous to a reduction from the extensive-form game into a
normal-form game for this case, played by the $t+k-(m-1)$ remaining rational and Byzantine players, in which
all rational players' dominating strategy is to bait the coalition,
regardless of what the rest are doing. Without this proof, Byzantine
players in the coalition could threat rational players to also bait if
they see them baiting, creating a deterrent and changing the
equilibrium of rational players into colluding to finalize the
disagreement.

We first show in Lemma~\ref{lem:can} that the \TRAP protocol
guarantees that no player can decide to join the baiting strategy
$\vv{\eta}$ and become a valid candidate for the winner consensus after
learning that another $m(k,t)$ players played $\vv{\eta}$: they must
take that decision before they know whether $m(k,t)$ other players
will play $\vv{\eta}$ or not.

\begin{lem}
  Let $n$ players play the associated game of the \TRAP protocol
$\vv{\sigma}$, out of which $k$ can be rational and $t$ Byzantine,
with $n>\frac{3}{2}k+3t$ and $n>2(k+t)$. Suppose a run in which a
coalition causes a disagreement on predecisions and players start the
BFTCR phase. Then, deviating player $i$ in the coalition cannot become
a valid candidate for the reward unless it commits to bait before it
learns that $m(k,t)$ other players commit to bait.
  \label{lem:can}
\end{lem}
\begin{proof}
  We show that if $m(k,t)$ rational players in the coalition
play the baiting strategy, becoming valid candidates to win the
reward, then the remaining $k+t-m(k,t)$ cannot obtain valid PoBs to
become candidates of the winner consensus after learning that $m(k,t)$
players become candidates.  Given that the non-baiting members of the
coalition are trying to finalize a disagreement, they will still split
non-deviating players into two partitions $A$ and $B$ for the BFTCR
protocol. Hence, we look at how many rational players must take part
in both partitions of the BFTCR protocol. Notice that $|A|+|B|+t+k\leq
n$, $|A|+k+t\geq n-t_0$ and $|B|+k+t\geq n-t_0$. Thus, analogous to
how we calculate $m$ in Lemma~\ref{lem:imp}, we have that $c\geq
(n-t_0)-\frac{n-t-k}{2}$ is the number of members of the coalition
that must participate in a partition for it to terminate deciding a
predecision, with $A\cap B=\emptyset$, as their predecisions
differ. We are interested in calculating $c-t$, the minimum number of
rational players out of these $c$ members of the coalition, this is
why we include as many Byzantine players as possible. Notice also that
we want to see how many rational players must take part in both
partitions, meaning that we are interested in
$c-t-\frac{k}{2}=(n-t_0)-\frac{n+t}{2}\geq m(k,t)$ for
$n>\frac{3}{2}k+3t$.

Hence, both partitions will include at least $m(k,t)$ repeated
rational players. What is left to prove is that if these $m(k,t)$
players commit to bait, then by the time they reveal their commitment,
the remaining players cannot collude to try and obtain PoBs to become
valid candidates of the winner consensus too. Since
$|A|+k+t\geq n-t_0$ and $|B|+k+t\geq n-t_0$, there are
$|D|\geq 2(n-t_0)-2k-2t$ correct players that delivered at least
$m(k,t)$ commitments to bait, for $|D|\leq |A|+|B|$, if these
$m(k,t)$ repeated rational players commit to bait. Notice that
$|D|\geq t_0+1$ for $n>2(k+t)$. Then, each of the $m(k,t)$ players can
wait for $n-t_0$ deliveries of the second reliable broadcast before
revealing their commitment by broadcasting their key without
compromising termination. Thus, we must calculate for which values of
$k$ and $t$ the remaining players cannot obtain PoBs to become valid
candidates, that is, for which values of $k$ and $t$ other players
that did not bait yet cannot include the new commitment to bait in
$t_0+1$ valid second reliable broadcasts. Since the remaining set of
correct players $C$ such that $|C|=n-t-k-|D|$ are
$|C|\geq n-k-t-(2(n-t_0)-2k-2t)$, we calculate for which values of $k$
and $t$ we have $|C|+k+t-t_0\leq t_0$, which results in
$n>2(k+t)$. This means that the $m(k,t)$ baiters can be sure that no deviating
player can commit to bait and win the reward without being a valid
candidate for the winner consensus.
\end{proof}

    We use the result from Lemma~\ref{lem:can} to prove \zeroloss and
\rdominance in Theorem~\ref{thm:usefulm}.

\begin{thm}[\zeroloss and \rdominance]
  
  \label{thm:usefulm}
  Let $\vv{\sigma}$ be the \TRAP protocol, executed by $n$ players of
which exactly $k$ are rational and $t$ Byzantine, for some values of $k,t$ satisfying
$n>\max(\frac{3}{2}k+3t,2(k+t))$. Let $\vv{\eta}$ be the strategy in
which $m(k,t)$ rational players reveal PoFs of the coalition if there
is a disagreement on predecisions. Then, $\vv{\eta}$ is a strong
baiting strategy if:
\begin{enumerate}
  \item 
% (i)
each player is required to deposit $\mathcal{L}=d\cdot
\mathcal{G}$, with $d>\frac{m(k,t)}{k(t_0-m(k,t)+1)}$, and
% (ii)
% , and
% 
% $\mathcal{G}$ the maximum total gain for the coalition that disagrees,    
  \item the baiting reward $\mathcal{R}$ is such that $\mathcal{R}=t_0\mathcal{L}$.
\end{enumerate}
\end{thm}
\begin{proof}
  Recall that the gain is split equally among all $k$ rational players
in the coalition $g=\mathcal{G}/k$. To guarantee \zeroloss, the sum of
losses from the coalition must always be equal or greater than the
reward given for the coalition to always lose funds while failing to
disagree, that is $t_0\mathcal{L}\geq \mathcal{R}\iff \mathcal{L}\geq
\frac{\mathcal{R}}{t_0}$.% . Since the % reward is always the same, we
% obtain a maximum deposit per player at a % minimum size of $k+t$, that
% is $k+t_0=t+1=\ceil{\frac{n-3}{3}-1}$

As a result, the baiting strategy $\vv{\eta}$ must strictly dominate the
strategy to disagree for rational players, even if a rational player
knows another $m-1$ other rational players also play the same strategy
$\vv{\eta}$ committing to bait. Since the probability of winning the bait between $m$
players is uniformly distributed $p(m)=\frac{1}{m}$ we have that the
utility for a player to play the baiting strategy knowing that another
$m(k,t)-1$ players are playing the same strategy is
$p(m(k,t))\mathcal{R}-q(m(k,t))\mathcal{L}$. If, instead, the player
disagrees then the player's utility is $\frac{\mathcal{G}}{k}$. As
such, and since Lemma~\ref{lem:can} shows that no rational player can
become a valid candidate to win the reward after learning that $m(k,t)$ other
players commit to bait, we obtain that $\vv{\eta}$ strictly dominates the
disagreeing strategy if
$p(m(k,t))\mathcal{R}-q(m(k,t))\mathcal{L}>\frac{\mathcal{G}}{k}$ and
replacing $\mathcal{R}$ by $t_0\mathcal{L}$, and $\mathcal{L}$ by
$d\mathcal{G}$ we obtain:

\begin{align*} &d\!>\!\Big(k\big(t_0p(m(k,t))-q(m(k,t))\big)\Big)^{-1}\!%
% \iff&d>\Big(k\big(\frac{t_0-m(k,t)+1}{m(k,t)}\big)\Big)^{-1},\\
%\iff
\Leftrightarrow\!&d\!>\!\frac{m(k,t)}{k(t_0-m(k,t)+1)}. % d>\frac{m}{(t+1-t)(t-m+1)}
\end{align*}

As for the reward, $t_0\mathcal{L}\geq \mathcal{R}$ for the slashed
deposits to always cover the reward, and thus we set
$t_0\mathcal{L}=d\mathcal{G}t_0=\mathcal{R}$.

Hence, $m(k,t)$ will play the baiting strategy (\rdominance) of which
one will be rewarded, and the reward will be paid with the deposits of
the fraudsters (\zeroloss). 
\end{proof}
Notice that any two values
$\mathcal{L}$ and $\mathcal{R}$ suffice if they satisfy
\inlineequation[eq:rl]{p(m(k,t))\mathcal{R}-q(m(k,t))\mathcal{L}>\frac{\mathcal{G}}{k}}, so that rational players prefer to bait than to disagree,
and \inlineequation[eq:rl2]{t_0\mathcal{L}\geq \mathcal{R}}, so that the reward is always less than the slashed deposits.
% Notice that any two values
% $\mathcal{L}$ and $\mathcal{R}$ suffice if they satisfy the following two equations:
% \begin{align}
%   p(m(k,t))\mathcal{R}-q(m(k,t))\mathcal{L}>&\frac{\mathcal{G}}{k}\label{eq:rl}\\
%   t_0\mathcal{L}\geq \mathcal{R}\label{eq:rl2}
% \end{align}

The key to these two equations lies in the trade-off between $\mathcal{R}$ and $\mathcal{L}$, that is: $\mathcal{R}$ must be sufficiently big compared to $\mathcal{L}$ so that players prefer to bait than to disagree (Equation~\ref{eq:rl}), but $\mathcal{R}$ must be sufficiently small compared to $\mathcal{L}$ so that the slashed deposits can always pay for the reward (Equation~\ref{eq:rl2}).

It is already possible to derive from Theorem~\ref{thm:usefulm}
results for the number of Byzantine players tolerated for $(k-t,t)$-robustness, given a deposit. That is, suppose that $\vv{\eta}$ only requires $m(k,t)=1$ rational player to satisfy agreement, and let $\mathcal{L}=d\cdot G$, then every coalition of size at least $t_0+1$ players has at least $k\geq t_0+1-t$ rational players, and thus the maximum amount of Byzantine players tolerated for $\epsilon$-$(k-t,t)$-robustness is $t<t_0+1-\frac{1}{t_0d}$. For example, let us set the deposit $\mathcal{L}=d\mathcal{G}$ to $d=\frac{1}{n}$, i.e., the total
deposit is $\mathcal{D}=\mathcal{L}\cdot n=\mathcal{G}$, and $n=100$, it follows that the
\TRAP protocol is $\epsilon$-$(k-t,t)$-robust and $t\leq 30$. 
If instead
$d=\frac{1}{3n}$, then $t\leq 24$. 
% \ARP{introduce this paragraph, or subsection, or even section}We only have left to prove that baiting the coalition is in fact a strong $(k,t,m(k,t))$-baiting strategy that $m(k,t)$ players will play, to guarantee \rdominance and \zeroloss. We show this in Theorem~\ref{thm:claim}.

Finally, we gather all results together in Theorem~\ref{thm:claim}, and Corollary~\ref{thm:everything}.
\begin{thm}
  Let $\vv{\sigma}$ be the \TRAP protocol, executed by $n$ players of
  which $k$ are rational and $t$ Byzantine, for all values $k,t$ satisfying
  $n>\max(\frac{3}{2}k+3t,2(k+t))$. Let $\vv{\eta}$ be the strategy in
  which $m(k,t)$ rational players reveal PoFs of the coalition if there
  is a disagreement on predecisions. Then, $\vv{\eta}$ is a strong
$(k,t,m(k,t))$-baiting strategy if:

  % Let $n$ players execute the \TRAP protocol, out
  % of which there can be a coalition of $k$ rational players and $t$
  % Byzantine players such that $n>2(k+t)$ and $n>\frac{3}{2}k+3t$.  
  % Suppose each of the following predicates holds:
  \begin{enumerate}
  \item 
  % (i) 
  Each player is required to deposit $\mathcal{L}=d\cdot \mathcal{G}$, where $d>\max_{(k,t)}\big(\frac{m(k,t)}{k(t_0-m(k,t)+1)}\big)$, and % , and $\mathcal{G}$ the maximum total gain for the coalition that disagrees
  \item  % (ii) 
  the baiting reward is $\mathcal{R}=t_0\mathcal{L}$.
  \end{enumerate}
      \label{thm:claim}
    \end{thm}
        \begin{proof}
    Theorem~\ref{thm:nhalf} uses the proof of \baitingagreement from
Lemma~\ref{thm:unavoidable} to show that if $m(k,t)$ play the
baiting strategy in the event of a disagreement on predecisions, then
the \TRAP protocol solves the \problem problem. Theorem~\ref{thm:usefulm}
shows that $m(k,t)$ will play the baiting strategy (\rdominance) of
which one will be rewarded, and the reward will be paid with the
deposits of the fraudsters (\zeroloss).

% As such, the \TRAP protocol solves \problem even in the presence
% of coalitions with up to $k$ rational players and $t_0$ Byzantine
% players.
Finally, we consider all possible values of $k$ and $t$ analogously to
Theorem~\ref{thm:usefulm}, deriving a value of $d$ that holds for all
possible values of $k$ and $t$:
$d>\max_{(k,t)}\big(\frac{m(k,t)}{k(t_0-m(k,t)+1)}\big)$.
\end{proof}

Notice that the greater the size of the coalition, the greater $d$
must be in order for the protocol to be
$\epsilon$-$(k,t)$-robust. However, for $n>\frac{3}{2}k+3t$ and
$n>2(k+t)$, since for every two rational players that join the coalition
one Byzantine must leave, the coalition that maximizes the total deposit 
%the maximum total required deposit
$\mathcal{D}=\mathcal{L}n=d\mathcal{G}n$ %takes place for 
is a coalition of $k=1$ rational player and
$t=t_0$ Byzantine players, and that means $d>\frac{1}{\ceil{\frac{n}{3}}-1}$. Corollary~\ref{thm:everything}
shows such particular robustness.  % We defer the proof to
% Appendix~\ref{app:proofs}.

% \TODO{put more examples?: whereas to support ... it must increase up to $d=\frac{3}{2n}$. Notice nevertheless that we cannot tolerate any more than $\frac{3}{2}k+3f<n$ and $k+f<n/2$}
\begin{cor}
  Let $\vv{\sigma}$ be the \TRAP protocol. Then $\vv{\sigma}$ is $\epsilon$-$(k,t)$-robust for the rational agreement problem for $n>\frac{3}{2}k+3t$ and $n>2(k+t)$ if the following predicates hold:
  \begin{enumerate}
  \item Each player is required to deposit $\mathcal{L}=d\cdot \mathcal{G}+\delta$, where $d =\frac{1}{\ceil{\frac{n}{3}}-1}$ % , $\mathcal{G}$ is the maximum total gain for the coalition that disagrees 
    and $\delta > 0$, and
  \item the baiting reward is $\mathcal{R}=t_0\mathcal{L}$.
  \end{enumerate}%   Let $\vv{\sigma}$ be the \TRAP protocol. Then $\vv{\sigma}$ is $\epsilon$-$(k,t)$-robust for the rational agreement problem for $n>\frac{3}{2}k+3t$ and $n>2(k+t)$ if (i) each player is required to deposit $\mathcal{L}=d\cdot \mathcal{G}+\delta$, where $d =\frac{1}{\ceil{\frac{n}{3}}-1}$ for $\delta > 0$, and (ii)
  % % \item
  %   the baiting reward is $\mathcal{R}=t_0\mathcal{L}$.
  % % \end{enumerate}
  \label{thm:everything}
\end{cor}
% \begin{proof}
%   This is a particular case of Theorem~\ref{thm:claim}.
% \end{proof}

Thus, there are two possible outcomes for the \TRAP protocol:
  \begin{itemize}[leftmargin=* ,wide=\parindent]
  \item if the coalition is made by so many rational players that deviating does not compensate the risk of losing their deposits, then the \TRAP protocol will provide agreement at predecision level without paying a reward $\mathcal{R}$, or
  \item if the coalition has enough Byzantine players to make the deviation into two predecisions profitable, then enough $m(k,t)$ rational players in the coalition will bait so that the disagreement on predecisions can safely be resolved and decided, and one rational player among the baiters will receive a reward $\mathcal{R}$, paid entirely by the deposits of the rest of the provably fraudulent players.
  \end{itemize}
     %\vspace{-0.4em}
    In both scenarios, the \TRAP protocol implements \problem, being $\epsilon$-$(k,t)$-robust for $n>\max\left(\frac{3}{2}k+3t,2(k+t)\right)$.

\section{Conclusion}
\label{sec-9}
%\TODO{recap, future work, conjectures}
In this paper, we showed that rational players help reduce the dependability on correct players to solve the Byzantine consensus problem.
To this end, we introduced a necessary and sufficient baiting strategy to solve the rational agreement problem---a variant of the Byzantine consensus problem that also tolerate rational players---under partial synchrony.
Based on this strategy, we also proposed \TRAP, a novel Byzantine consensus protocol among $n>\max\big(\frac{3}{2}k+3t,2(k+t)\big)$ players, where $k$ players are rational and $t$ are Byzantine. 
This protocol tolerates the coordinated deviations of up to $k$ rational players and $t$ Byzantine players, solving consensus in the presence of less than $2n/3$ correct players.
%\TODO{remove?: and guarantees that $n-t$ players agree despite the arbitrary behaviors of $t\leq t$ Byzantine players. }
As future work, it would be interesting to explore whether our bound
$n>\max\big(\frac{3}{2}k+3t,2(k+t)\big)$ is tight, and to consider the impact
of non-negligible costs of computation and communication. We are
currently working at reducing our bound to $n>\frac{3}{2}k+3t$ with
a novel VSS scheme.
%can hardly be improved without further
%assumptions. 

% Emacs 26.3 (Org mode 8.2.10)
\bibliography{game-theory}

\clearpage
\appendix
\section{Extended Example Figure}
\label{sec:exfig}

Figure~\ref{fig:diag2} depicts
a slightly extended version of the execution example of Figure~\ref{fig:diag}.
%a slightly extended example
%execution, to detail the execution more in depth compared to
%Figure~\ref{fig:diag}. 
Similarly to Figure~\ref{fig:diag}, the execution
starts with $k+t$ Byzantine and rational players causing a disagreement
on predecisions.  However, now we detail further how the $m$ baiters
prevent termination of the BFTCR protocol. In particular, by not
committing to a value in the first reliable broadcast of BFTCR, the
$m$ baiters can prevent players in $A$ and in $B$ from terminating in
any of the two partitions. Thus, the $m$ baiting players wait till
they receive certificates from players in $A$ and in $B$ in order to
construct PoFs. Then, they wait till they deliver enough values from
the second group reliable broadcasts from players in partitions $A$
and $B$ that guarantee that no other Byzantine or rational player can
become a valid candidate once they reveal that they are baiting (as we
showed in the proof of Lemma~\ref{lem:can}). At this point, the $m$
players reveal their PoFs by sending the decryption key to their
commitment. Then, players in $A$ and $B$ can resolve their
disagreement on predecisions, choose a winner of the reward from among
the $m$ valid candidates at random, and punish the rest of deviating
players.

\begin{figure}[htp]
    \hspace{-6em}
    \pgfsetlayers{background,main,foreground}
    \begin{tikzpicture}[node distance=1cm,auto,>=stealth']
      \tikzstyle{every node}=[font=\small]
      \node[] (byz) {$t$};
      % \node[yshift=1.5em,xshift=-1.5em] (byzauxbyz) {Byzantines};
         \node[right = of byz,xshift=1.5em] (km) {$k-m$};
   \node[right = of km,xshift=1.5em] (m) {$m$};
   \node[right = of m,xshift=4.5em] (a) {$\color{blue}A$};
   \node[right = of a,xshift=4.5em] (b) {$\color{red}B$};
   %       \node[right = of byz,xshift=-2em] (km) {$k-m$};
   % \node[right = of km,xshift=-2em] (m) {$m$};
   % \node[right = of m] (a) {$\color{blue}A$};
   % \node[right = of a,xshift=1em] (b) {$\color{red}B$};
   % \node[right = of byz,xshift=-2.5em] (km) {$k-m$};
   % \node[right = of km,xshift=-2.5em] (m) {$m$};
   % \node[right = of m,xshift=2em] (a) {$A$};
   % \node[right = of a,xshift=2em] (b) {$B$};
   \draw [decorate, 
    decoration = {calligraphic brace,
        raise=5pt,
        amplitude=5pt}] (a.west) --  (b.east)
      node[pos=0.45,above=10pt,black]{$A\cap B=\emptyset$, partition of correct};

      \draw [decorate, 
    decoration = {calligraphic brace,
        raise=5pt,
        amplitude=5pt}] (byz.west)+(-0.2,0) -- (0.3,0)+(byz.east)
      node[pos=0.2,above=10pt,black]{Byzantines};

      \draw [decorate, 
      decoration = {calligraphic brace,
        raise=5pt,
        amplitude=5pt}] (km.west) -- (m.east)
      node[pos=0.5,above=10pt,black,align=center]{$k$ rationals, of which $m$ bait};

   \def\vara{4}
   \node[below of=byz, node distance=\vara cm] (byz_ground) {};
   \node[below of=km, node distance=\vara cm] (km_ground) {};
   \node[below of=m, node distance=\vara cm] (m_ground) {};
   \node[below of=a, node distance=\vara cm] (a_ground) {};
   \node[below of=b, node distance=\vara cm] (b_ground) {};
   
   \def\varb{0.22}
   \draw (byz) -- ($(byz)!\varb!(byz_ground)$);
   \draw (km) -- ($(km)!\varb!(km_ground)$);
   \draw (m) -- ($(m)!\varb!(m_ground)$);
   \draw (a) -- ($(a)!\varb!(a_ground)$);
   \draw (b) -- ($(b)!\varb!(b_ground)$);      

   \def\varc{0.36}
   \def\vard{1.595}
   \draw ($(byz)!\varc!(byz_ground)$) -- ($(byz)!\vard!(byz_ground)$);
   \draw ($(km)!\varc!(km_ground)$) -- ($(km)!\vard!(km_ground)$);
   \draw ($(m)!\varc!(m_ground)$) -- ($(m)!\vard!(m_ground)$);
   \draw ($(a)!\varc!(a_ground)$) -- ($(a)!\vard!(a_ground)$);
   \draw ($(b)!\varc!(b_ground)$) -- ($(b)!\vard!(b_ground)$);

   % \path  ($(byz)!0.22!(byz_ground)$)-- ($(byz)!0.30!(byz_ground)$)node [red, font=\Huge, midway, center] {$\dots$};
   \def\vare{0.25}
   \def\varf{0.36}
   \draw[loosely dotted, line width=0.5mm] ($(byz)!\vare!(byz_ground)$)-- ($(byz)!\varf!(byz_ground)$);
   \draw[loosely dotted, line width=0.5mm] ($(km)!\vare!(km_ground)$)-- ($(km)!\varf!(km_ground)$);
   \draw[loosely dotted, line width=0.5mm] ($(m)!\vare!(m_ground)$)-- ($(m)!\varf!(m_ground)$);
   \draw[loosely dotted, line width=0.5mm] ($(a)!\vare!(a_ground)$)-- ($(a)!\varf!(a_ground)$);
   \draw[loosely dotted, line width=0.5mm] ($(b)!\vare!(b_ground)$)-- ($(b)!\varf!(b_ground)$);
   
   % % \draw (m) -- (m_ground);
   % % \draw (a) -- (a_ground);
   % % \draw (b) -- (b_ground);

   \def\varg{0.05}
   \def\varh{0.2}
   \draw[-latex',color=blue,dashed] ($(byz)!\varg!(byz_ground)$) -- node[below,near start]{% \scriptsize $\lit{dec}(v_A)$
   } ($(a)!\varh!(a_ground)$);
   \draw[-latex',color=red,densely dotted] ($(byz)!\varg!(byz_ground)$) -- node[below,near start]{% \scriptsize $\lit{dec}(v_B)$
   } ($(b)!\varh!(b_ground)$);

   \draw[-latex',color=blue,dashed] ($(km)!\varg!(km_ground)$) -- node[below,near start]{% \scriptsize $\lit{dec}(v_A)$
   } ($(a)!\varh!(a_ground)$);
   \draw[-latex',color=red,densely dotted] ($(km)!\varg!(km_ground)$) -- node[below,near start]{% \scriptsize $\lit{dec}(v_B)$
   } ($(b)!\varh!(b_ground)$);

   \draw[-latex',color=blue,dashed] ($(m)!\varg!(m_ground)$) -- node[below,near start]{% \scriptsize $\lit{dec}(v_A)$
   } ($(a)!\varh!(a_ground)$);
   \draw[-latex',color=red,densely dotted] ($(m)!\varg!(m_ground)$) -- node[below,near start]{% \scriptsize $\lit{dec}(v_B)$
   } ($(b)!\varh!(b_ground)$);

   \def\vari{0.23}
   \def\varj{0.20}
   \def\vark{0.37}
   \def\varm{0.5}
   \node[xshift=2.4em] (a_predec) at ($(a)!\vari!(a_ground)$) {\ding{192~}{\color{blue}$\lit{predec}(v_A)$}};
   \node[xshift=2.4em] (b_predec) at ($(b)!\vari!(b_ground)$) {\ding{192~}{\color{red}$\lit{predec}(v_B)$}};
   \node[xshift=-6em,align=center,style={rectangle,draw},minimum width=11em] (exp1) at ($(byz)!\varj!(byz_ground)$) {\footnotesize \ding{192}~$t+k$ lead $A$ and $B$ to \\\footnotesize disagreement on predecisions};
      % \draw[loosely dotted, line width=0.5mm] ($(b)!0.235!(b_ground)$)-- ($(b)!0.285!(b_ground)$);

   \draw[-latex',color=blue,dashed] ($(byz)!\vark!(byz_ground)$) -- node[below,near start]{% \scriptsize $\lit{dec}(v_A)$
   } ($(a)!\varm!(a_ground)$);
   \draw[-latex',color=red,densely dotted] ($(byz)!\vark!(byz_ground)$) -- node[below,near start]{% \scriptsize $\lit{dec}(v_B)$
   } ($(b)!\varm!(b_ground)$);

   \draw[-latex',color=blue,dashed] ($(km)!\vark!(km_ground)$) -- node[below,near start]{% \scriptsize $\lit{dec}(v_A)$
   } ($(a)!\varm!(a_ground)$);
   \draw[-latex',color=red,densely dotted] ($(km)!\vark!(km_ground)$) -- node[below,near start]{% \scriptsize $\lit{dec}(v_B)$
   } ($(b)!\varm!(b_ground)$);
   \node[align=center,xshift=0.5em] (dingaux1) at ($(a)!\varm!(a_ground)$) {\ding{193}};
   \node[align=center,xshift=0.5em] (dingaux1) at ($(b)!\varm!(b_ground)$) {\ding{193}};
   \ding{193}~
   % % \node[xshift=-4em,align=center] (exp1) at ($(byz)!0.09!(byz_ground)$) {$t+k$ lead to\\ disagreement\\ on predecisions};
   % % \node[align=center] (exp1) at ($(m)!0.25!(m_ground)$) {\Huge$\vdots$};

   \def\varn{0.5}
   \def\varo{0.8}
   \node[xshift=-6em,align=center,style={rectangle,draw},minimum width=11em] (exp2) at ($(byz)!\varn!(byz_ground)$) {\ding{193}~\footnotesize $A$ and $B$ cannot terminate\\\footnotesize without hearing from at least\\\footnotesize $1$ of the $m$ baiters};

   \draw[-latex'] ($(a)!\varj!(a_ground)$) -- node[below,pos=0.5, sloped]{ \scriptsize $\ms{certificate}_{{\color{blue}(v_A)}}$
   } ($(m)!\varo!(m_ground)$);

   \draw[-latex'] ($(b)!\varj!(b_ground)$) -- node[below,near end, sloped]{ \scriptsize $\ms{certificate}_{{\color{red}(v_B)}}$
   } ($(m)!\varo!(m_ground)$);

   \node[align=center,xshift=-0.5em] (dingaux2) at ($(m)!\varo!(m_ground)$) {\ding{194}};

   \node[xshift=-6em,yshift=-0.5em,align=center,style={rectangle,draw},minimum width=11em] (exp3) at ($(byz)!\varo!(byz_ground)$) {\ding{194}~\footnotesize $m$ baiters wait to receive \\\footnotesize certificates to construct PoFs};

   \def\varp{0.85}
   % \node[align=center,yshift=0.2em,xshift=0.2em] (mpofs) at ($(m)!\varp!(m_ground)$) {$\ms{pofs}_{({\color{blue}v_A},{\color{red}v_B})}$};

   \def\varq{0.88}
   \def\varr{1.1}
   \draw[-latex'] ($(m)!\varq!(m_ground)$) -- node[below,pos=0.5, sloped]{ \scriptsize $RB^1(\ms{enc\_pofs}_{({\color{blue}v_A},{\color{red}v_B})})$
   } ($(a)!\varr!(a_ground)$);

   \draw[-latex'] ($(m)!\varq!(m_ground)$) -- node[above,pos=0.7, sloped]{ \scriptsize $RB^1(\ms{enc\_pofs}_{({\color{blue}v_A},{\color{red}v_B})})$
   } ($(b)!\varr!(b_ground)$);

   \node[align=center,xshift=0.5em] (dingaux2) at ($(a)!\varr!(a_ground)$) {\ding{195}};
   \node[align=center,xshift=0.5em] (dingaux2) at ($(b)!\varr!(b_ground)$) {\ding{195}};

   \node[xshift=-6em,yshift=-0.5em,align=center,style={rectangle,draw},minimum width=11em] (exp4) at ($(byz)!\varr!(byz_ground)$) {\ding{195}~\footnotesize $m$ baiters commit to\\\footnotesize the encrypted PoFs};

   \def\vars{1.2}
   \def\vart{1.35}
   \draw[-latex'] ($(a)!\vars!(a_ground)$) -- node[above,pos=0.5, sloped]{ \scriptsize $RB^2_A$
   } ($(m)!\vart!(m_ground)$);

   \draw[-latex'] ($(b)!\vars!(b_ground)$) -- node[below,pos=0.7, sloped]{ \scriptsize $RB^2_B$
   } ($(m)!\vart!(m_ground)$);
   \node[align=center,xshift=0.5em,yshift=-0.3em] (dingaux2) at ($(m)!\vart!(m_ground)$) {\ding{196}};

   \node[xshift=-6em,yshift=-0.5em,align=center,style={rectangle,draw},minimum width=11em] (exp4) at ($(byz)!\vart!(byz_ground)$) {\ding{196}~\footnotesize $m$ baiters wait to confirm\\\footnotesize they are valid candidates};
   
   \def\varu{1.4}
   \def\varv{1.55}
   \draw[-latex'] ($(m)!\varu!(m_ground)$) -- node[below,pos=0.5, sloped]{ \scriptsize reveal $\ms{key}$
   } ($(a)!\varv!(a_ground)$);

   \draw[-latex'] ($(m)!\varu!(m_ground)$) -- node[above,pos=0.7, sloped]{ \scriptsize reveal $\ms{key}$
   } ($(b)!\varv!(b_ground)$);

   \node[align=center,xshift=0.5em] (dingaux2) at ($(a)!\varv!(a_ground)$) {\ding{197}};
   \node[align=center,xshift=0.5em] (dingaux2) at ($(b)!\varv!(b_ground)$) {\ding{197}};
   
   \node (aux) at ($(byz)!0.5!(b)$) {};
   \node (aux_ground) at ($(byz_ground)!0.5!(b_ground)$) {};
   \node[align=center,minimum width=35em,xshift=-1em,style={rectangle,draw},yshift=-0.7em] at ($(aux)!1.75!(aux_ground)$){\ding{197}~partitions of correct players $A$ and $B$ $(i)$ discover disagreement, \\ $(ii)$ select a winner of the reward at random out of the $m$ baiters,\\$(iii)$ punish the rest of $t+k-1$ deviants, and \\$(iv)$ resolve the disagreement on predecisions to agree on decision};
 \end{tikzpicture}
 \caption{Extended example execution of the \Huntsman protocol. First, \ding{192}~ all $t$ Byzantine and $k$ rational players collude to cause a disagreement on the output of the accountable consensus protocol, resulting in $A$ and $B$ predeciding different outputs. Then, \ding{193}~ $m$ of the $k$ rational players decide to bait while executing the BFTCR protocol, preventing $A$ and $B$ from deciding their disagreeing predecisions. As such, \ding{194}~the $m$ baiters wait until they receive proof of the disagreement on predecisions, to then \ding{195}~commit to the encrypted PoFs. Finally, \ding{196}~once they deliver as many second reliable broadcast from $A$ and $B$ as possible confirming that correct players delivered their PoFs encrypted commitment, then \ding{197} the $m$ baiters prove the disagreement revealing the proofs-of-fraud in the BFTCR protocol. Hence, neither $A$ nor $B$ decide their conflicting predecisions, but instead reward one of the $m$ baiters, punish the rest of $t+k-1$ players responsible for the disagreement on predecisions, and resolve the disagreement, deciding one of $v_A$ or $v_B$, or, depending on the application, merging both.}
 \label{fig:diag2}
  \end{figure}
\section{Discussion: paying a reward at no cost to non-deviants.}  One might think that
implementing a baiting strategy with a reward and deposits might not
be enough: we need to discourage coalitions from actually playing the
baiting strategy, since the system would have to pay the reward
$\mathcal{R}$, and thus the coalition can effectively steal some funds
from the system. However, if the system can use the deposited amount
$\mathcal{L}$ from at least $t_0$ certified fraudsters in the coalition
to pay for the baiting reward $\mathcal{R}$, then the system does not
lose any funds (\zeroloss), while obtaining agreement
(\baitingagreement).

Furthermore, notice that if the coalition consists entirely of
rational players then they do not actually play this strong baiting
strategy since, by the definition of strong baiting strategy, they all
individually lose more than they can gain from deviating. Even if the
presence of Byzantine players leads to a baiter being paid, agreement
will still be guaranteed at no cost to non-deviating players. This
leaves the open question of how likely it is that Byzantine players
with unexpected utilities but possibly with the goal to break the system would be
interested in giving their funds for free to rational players, if it
does not cause some damage on non-deviating players or on the system
itself. In other words, with a more refined, realistic modelling of
Byzantine players, it is very likely that the very correctness of the
\TRAP protocol will be enough of a deterrent from deviating, which would lead
to agreement directly at the predecision level.
\end{document}